\documentclass[11pt]{imsart}
\usepackage{amsmath,amssymb,fullpage,graphicx,color,bm}
\let\hat\widehat
\let\tilde\widetilde

\usepackage[authoryear,round]{natbib}

\newtheorem{theorem}{Theorem}
\newtheorem{lemma}[theorem]{Lemma}

\newtheorem{corollary}[theorem]{Corollary}

\newenvironment{proof}{{\bf Proof.}}{$\Box$}

\newcommand{\E}{\mbox{$\mathbb{E}$}}

\usepackage{multirow}
\usepackage{sectsty}
\definecolor{myblue}{RGB}{50,50,150}
\allsectionsfont{\color{myblue}\bfseries\sffamily}

\usepackage{longtable}
\usepackage{rotating}
\usepackage{array}
\usepackage{multicol}
\usepackage{multirow}

\begin{document}

\begin{center}
\textcolor{myblue}{
\textsf{\textbf{\Large Decorrelated Variable Importance}}}\\
Isabella Verdinelli and Larry Wasserman\\
November 19 2021
\end{center}

\begin{quote}
{\em
Because of the widespread use of black box prediction methods
such as random forests and neural nets,
there is renewed interest
in developing methods for quantifying variable importance as
part of the broader goal of
interpretable prediction.
A popular approach is to
define a variable importance parameter --- known as
LOCO (Leave Out COvariates) ---
based on dropping covariates from a
regression model. 
This is essentially a nonparametric version of $R^2$.
This parameter is very general
and can be estimated nonparametrically,
but it can be hard to interpret
because it is affected by correlation between covariates.
We propose a method for mitigating the effect of correlation
by defining a modified version of LOCO.
This new parameter is difficult to estimate nonparametrically, but
we show how to estimate it using semiparametric models.}
\end{quote}

\section{Introduction}

Due to the increasing popularity of
black box prediction methods like
random forests and neural nets,
there has been renewed interest
in the problem
of quantifying variable importance in
regression.
Consider predicting $Y\in\mathbb{R}$
from covariates $(X,Z)$
where
$X\in\mathbb{R}^g$
and
$Z\in\mathbb{R}^h$.
We have
separated the covariates
into $X$ and $Z$
where $X$ represents the covariates
whose importance we wish to assess.
In what follows,
we let
$U=(X,Z,Y)$ 
denote all the variables.
Define
$\mu(x,z) = \E[Y|X=x,Z=z]$
so that
$$
Y = \mu(X,Z) + \epsilon
$$
where
$\E[\epsilon|X,Z]=0$.

A popular measure of the importance of $X$ is
\begin{equation}
\psi_L = \E[(\mu(Z) - \mu(X,Z))^2] = \E[(Y-\mu(Z))^2] - \E[(Y - \mu(X,Z))^2].
\end{equation}
where $\mu(Z) = \E[Y|Z=z]$.
Up to scaling, $\psi_L$ is a nonparametric version of the usual $R^2$ from standard regression.
This was called LOCO ({\bf L}eave {\bf O}ut {\bf CO}variates) in
\cite{lei2018distribution} and \cite{rinaldo2019bootstrapping}
and has been further studied recently in
\cite{williamson2021nonparametric},
\cite{williamson2020unified} and
\cite{zhang2020floodgate}.
The parameter $\psi_L$ is appealing because it is very general
and easy to interpret.
But it suffers from some problems. In particular,
the value of $\psi_L$ depends on the correlation between $X$ and $Z$.
When $X$ and $Z$ are highly correlated,
$\psi$ will be near 0 since removing $X$ has little effect.
In some applications, this might be undesirable as it obscures interpretability.
We refer to this problem as {\em correlation bias}.
Another, more technical problem with LOCO, is its quadratic nature
which causes some issues when constructing confidence intervals.

In this paper,
we define a modified version of $\psi_L$ 
denoted by $\psi_0$ that is
invariant to the correlation between $X$ and $Z$.
There is a tradeoff:
the modified parameter $\psi_0$ is free from correlation bias
but it is more difficult to estimate than $\psi_L$.
In a sense, we remove the correlation from the estimand 
at the expense of larger confidence intervals.
This is similar to estimating a coefficient in a linear regression
where the value of the regression coefficient does not depend on the correlation
between $X$ and $Z$
while the width of the confidence interval does.
To reduce the difficulties
in estimating $\psi_0$,
we approximate $\mu(x,z)$
with the semiparametric model
$\mu(x,z) = \beta(z)^T x + f(z)$.

\bigskip

{\bf Related Work.}
Assessing variable importance
is an active area of research.
Recent papers on LOCO include
\cite{lei2018distribution,rinaldo2019bootstrapping,
williamson2021nonparametric,
williamson2020unified,
zhang2020floodgate}.
Another approach is to use
derivatives of the regression function as
suggested in
\cite{samarov1993exploring},
and has received renewed attention in the
machine learning literature
\citep{ribeiro2016model}.
There has been a surge of interest in
an approach based on Shapley values, see for example,
\cite{messalas2019model,
aas2019explaining,
lundberg2016unexpected,
covert2020understanding,
fryer2020shapley,
covert2020improving,
israeli2007shapley,
benard2021shaff}.
We discuss derivatives and Shapley values in Section
\ref{section::other}.
Another paper that uses semiparametric models
for intepretability is
\cite{sani2020semiparametric}
but that paper does not focus on
variable importance.

\bigskip

{\bf Paper Outline.}
In Section \ref{section::issues}
we describe some issues related to LOCO
and this leads us to define
a few modified versions of the parameter.
In Section \ref{section::inference}
we discuss inference for the parameters.
Section \ref{section::simulations}
contains some simulation studies.
Section \ref{section::other}
discusses other issues and other measures of
variable importance.
In Section, \ref{section::balancing}
we introduce a different approach based on covariate balancing.
A concluding discussion is in 
Section \ref{section::conclusion}.
Technical details and proofs are in an appendix.

\section{Issues With LOCO}
\label{section::issues}

The parameter $\psi_L$ is general
and it is easy to obtain point estimates for it;
see Section \ref{section::psi}.
But it does have two shortcomings
which we now discuss.

\subsection{Issue 1: Inference For Quadratic Functionals}

The first, and less serious issue, is that
$\psi_L$ is a quadratic parameter
and it is difficult to get confidence intervals for
quadratic parameters because their
limiting distribution and rate of convergence change
as $\psi_L$ approaches 0.
This is actually a common problem but it receives little attention.
Many other parameters have this problem, including
distance correlation \citep{szekely2007measuring},
RKHS correlations \citep{sejdinovic2013equivalence}
and kernel two-sample statistics
\citep{gretton2012kernel}
among others.

To illustrate,
consider the following toy example.
Let $Y_1,\ldots, Y_n \sim N(\mu,\sigma^2)$
and consider estimating
$\psi = \mu^2$ with
$\hat\psi = \overline{Y}_n^2$.
When $\mu\neq 0$,
we have
$\sqrt{n}(\hat\psi - \psi)\rightsquigarrow N(0,\tau^2)$
for some $\tau^2$.
When $\mu=0$,
$\hat\psi \sim \sigma^2 \chi_1^2/n$.
When
$\mu$ is close to 0, its distribution is neither Normal nor chi-squared,
and the rate of convergence can be anything between
$1/n$ and $1/\sqrt{n}$.

More generally,
when dealing with a quadratic functional $\psi$,
it is often the case that
an estimator $\hat\psi$
converges to a Normal at a $n^{-1/2}$ rate
when $\psi \neq 0$
but at the null, where $\psi=0$,
the influence function for the parameter vanishes,
the rate becomes $n^{-1}$
and the limiting distribution is typically
a combination of $\chi^2$ random variables.
Near the null, we get behavior in between these two cases.
A valid confidence interval $C_n$ should satisfy
$P(\psi_n \in C_n)\to 1-\alpha$
even if $\psi_n$ is allowed to change with $n$.
In particular, we want to allow $\psi_n\to 0$.
Finding a confidence interval with this
uniformly correct coverage,
with length $n^{-1/2}$ away from the null
and length $n^{-1}$ at the null
is, to the best of our knowledge, an unsolved problem.

Our proposal is 
to construct a conservative confidence interval
that does not have length $O(1/n)$ at the null.
We replace the standard error $se$ of $\hat\psi$
with
$\sqrt{{se^2} + c^2/n}$
where $c$ is a constant.
We take $c = ({\rm Var}[Y])^2$ 
to put the quantity on the right scale, but other constants could be used.
This leads to valid confidence intervals
but they are conservative near the null
as they shrink at rate $n^{-1/2}$
instead of $n^{-1}$.

We are only aware of two other attempts to address this issue.
Both involve expanding the width of the confidence interval to be $O(n^{-1/2})$.
\cite{dai2021significance}
added noise of the form $c Z/\sqrt{n}$ 
to the estimator,
where $Z\sim N(0,1)$. 
They choose
$c$ by permuting the data many times
and finding a $c$ that gives good coverage under the simulated permutations.
However, this is computationally expensive
and adding noise seems unnecessary.
\cite{williamson2020unified}
deal with this problem by
writing $\psi$ as a sum of two
parameters
$\psi = \psi_1 + \psi_2$
such that neither $\psi_1$ nor $\psi_2$ vanish
when $\psi=0$.
Then, they estimate $\psi_1$ and $\psi_2$
on separate splits of the data.
This again amounts to adding noise of size $O(1/\sqrt{n})$.

All three approaches are basically the same;
they have the effect of expanding
the confidence interval by $O(n^{-1/2})$
which maintains validity at the expense
of efficiency at the null.
Our approach has the virtue of being simple and fast.
It does not require adding noise,
extra calculations or doing an extra split of the data.

To see that 
expanding the standard error does lead to an interval with correct
coverage,
let $\hat\psi$ denote an estimator of a parameter $\psi_n$
which we allow to change with $n$.
We are concerned with the case were the
bias $b_n$ satisfies $b_n = o(n^{-1/2})$
and the variance $v_n$ satisfies
$v_n = o(1/n)$.
(The variance would be of order $1/n$ in the non-degenerating case.)
Then, by Markov's inequality,
the non-coverage of
the interval
$\hat\psi_n \pm z_{\alpha/2}\sqrt{se^2 + c^2/n}$ is
\begin{align*}
P\Bigl(|\hat\psi_n - \psi_n| > z_{\alpha/2}\sqrt{se^2 + c^2/n}\Bigr) &\leq
P\Bigl(|\hat\psi_n - \psi_n| > z_{\alpha/2}\sqrt{c^2/n}\Bigr)\\
& \leq
\frac{n}{c z_{\alpha/2}^2} \E[ | \hat\psi_n - \psi_n|^2] =
\frac{n}{c z_{\alpha/2}^2} (b_n^2 + v_n) = o(1).
\end{align*}

\subsection{Issue 2: Correlation Bias}

The second and more pernicious problem 
is that $\psi_L$ depends on the correlation between $X$ and $Z$.
In particular, if $X$ and $Z$ are highly correlated,
then $\psi_L$ will typically be close to 0.
We call this, {\em correlation bias}.
There may be applications where this is acceptable.
But in some cases we may want to alleviate this bias
and that is the focus of this paper.

To appreciate the effect of correlation bias,
consider the linear model
$Y = \beta X + \theta Z + \epsilon$.
In this case,
a natural measure of variable importance is $\beta$
which is unaffected by correlation between $X$ and $Z$.
The standard error of the estimate
$\hat\beta$ is affected by the correlation
but the estimand itself is not.
For this model,
$\psi_L = \beta^2 \gamma^2$
where
$\gamma^2 = \E[(X-\nu(Z))^2]$ and
$\nu(z) = \E[X|Z=z]$.
This makes it clear that
$\psi_L\to 0$
as $X$ and $Z$ become more correlated.
The same fate befalls the partial correlation $\rho$ between
$Y$ and $X$ which in this model is
$\rho = (1 + \frac{\beta^2 \sigma^2}{\gamma^2})^{-1/2}$
where
$\sigma^2 = {\rm Var}[\epsilon]$.
Again, $\rho \to 0$ as
$\gamma \to 0$.

To deal with this problem,
we define a modified LOCO parameter
$\psi_0$ which is unaffected by the
dependence between $X$ and $Z$.
Let
$p_0(x,y,z) = p(y|x,z)p(x)p(z)$.
Then
$p_0$ is the distribution
that is closest to
$p$ in Kullback-Leibler distance
subject to making $X$ and $Z$ independent.
We define
\begin{equation}\label{eq::psi*}
\psi_0 = \E_0[ (\mu_0(X,Z) - \mu_0(Z))^2 ].
\end{equation}
A simple calculation shows that
$\mu_0(z) = \E_0[Y|Z=z] = \int\mu(x,z) p(x) dx$
and so
\begin{equation}\label{eq::psi0}
\psi_0 = \int (\mu_0(x,z) - \mu_0(z))^2 p(x)p(z) dx dz.
\end{equation}
We can think of $\psi_0$
as a counterfactual quantity
answering the question:
what would the change in $\mu(X,Z)$ be
if we dropped $X$ and had $X$ and $Z$ been independent.

This parameter completely eliminates the correlation
bias but, as we show in our simulations, it can be hard to
get an accurate estimate of $\psi_0$.
In particular, nonparametric confidence intervals are wide.
A simple, but somewhat ad-hoc solution,
is to first remove $Z_j's$ that are highly correlated with $X$.
That is,
define $\psi_1 = \E[ ( \mu(V) - \mu(X,V))^2 ]$
where $V=(Z_j:\ |\rho(X,Z_j)| \leq t)$ for some $t$
where $\rho$ is a measure of dependence.

The main solution we propose is to
use the semiparametric model
$\mu(x,z) = x^T\beta(z) + f(z)$.
Under this model,
one can show that
$\psi_0$ takes the form
${\rm tr}\left(\Sigma_X \E[\beta(Z)\beta(Z)^T]\right)$
where
$\Sigma_X = {\rm Var}[X]$.
(See appendix 8.4 for details).
However, this parameter is still difficult to estimate
so we propose the following two simpler models.
First, 
let $\mu(x,z) = \beta^T x + f(z)$.
Then $\psi_0$  becomes
\begin{equation}\label{eq::psi2}
\psi_2 = \beta^T \Sigma_X \beta.
\end{equation}
The second model is
\begin{equation}\label{eq::model2}
\mu(x,z) = \beta^T x + \sum_{j}\sum_j \gamma_{jk} x_j z_k + f(z).
\end{equation}
In Section \ref{section::est3} we show that $\psi_0$ then becomes
\begin{equation}\label{eq::psi3}
\psi_3 = \theta^T \Omega \theta
\end{equation}
where
$$
\theta = 
\Biggl\{
\E\bigl[\tilde Z \tilde Z^T \otimes (X-\nu(Z))(X-\nu(Z))^T\bigr] \Biggr\}^{-1}
\E\Biggl[ \biggl(Y-\mu(Z)\biggr)\  \biggl(\tilde Z \otimes (X-\nu(Z))\biggr)\Biggr].
$$
$\nu(z) = \E[X|Z=z]$,
$\tilde Z = (1,Z)$
and
$$
\Omega = \Sigma_X \otimes \E[ \tilde Z \tilde Z^T] =
\Sigma_X \otimes \left(
\begin{array}{cc}
1 & m_Z^T \\
m_Z & \Sigma_Z + m_Z m_Z^T
\end{array}
\right),
$$
$m_Z=\E[Z]$ and
$\Sigma_Z = {\rm Var}[Z]$.
Table \ref{table::psi}
summarizes the expressions for the parameters.

\bigskip

{\bf Remark:}
{\em In all the above definitions,
we can replace
$X$ with
$b(X) = (b_1(X),\ldots, b_k(X))$
for a given set of basis functions
$b_1,\ldots, b_k$
to make the model more flexible.
For example,
we can take
$b(X) = (X,X^2,X^3)$ or
an orthogonalized version of the polynomials,
which is what we use in several of our examples.}

In these semiparametric models,
we can estimate
the nuisance functions
$\nu(z) = \E[X|Z=z]$ and $\mu(z)$
either nonparametrically or parametrically.

\begin{table}
\fbox{\parbox{5.5in}{
\begin{center}
\begin{tabular}{ll}
$\psi_0 = \int\int (\mu(x,z)-\mu_0(z))^2 p(x)p(z) dx dz$ & \hspace{1em}
$\psi_1 = \E[ (\mu(X,V)-\mu(V))^2]$\\
$\psi_2 = \beta^T \Sigma_X \beta$ & \hspace{1em}
$\psi_3 = \theta^T \Omega \theta$\\
\\
\rule{2in}{.1mm}  & \hspace{-7em}\rule{3in}{.1mm}\\
\\
$\mu_0(z) = \int \mu(x,z) p(x)dx$ &
$V = (Z_j:\ |\rho(X,Z_j)| \leq t)$ \\
$\tilde Z^T = (1,Z^T)$ &
$\Omega = \Sigma_X \otimes
\left[
\begin{array}{cc}
1 & m_Z^T \\
m_Z & \Sigma_Z + m_Z m_Z^T\\
\end{array}
\right] $\\
$\beta = \E[ (Y-\mu(Z))(X-\nu(Z))]/\E[(X-\nu(Z))^2]$\\
\multicolumn{2}{l}{
$\theta = 
\Biggl\{
\E[\tilde Z \tilde Z^T \otimes (X-\nu(Z))(X-\nu(Z))^T] \Biggr\}^{-1}
\E\Biggl[ \biggl(Y-\mu(Z)\biggr)\  \biggl(\tilde Z \otimes (X-\nu(Z))\biggr)\Biggr]$}\\
\end{tabular}
\end{center}
}}
\bigskip
\caption{\bf Summary of Decorrelated Parameters}.
\label{table::psi}
\end{table}

\section{Inference}
\label{section::inference}

In this section we discuss estimation
of $\psi\in\{\psi_L,\psi_0,\psi_1,\psi_2,\psi_3\}$.
For $\psi_0, \psi_2$ and $\psi_3$
we use one-step estimation
which we now briefly review. 
See \cite{hines2021demystifying}
for a recent tutorial on one-step estimators.
Let $\psi(\gamma)$ be a parameter
with efficient influence function $\phi(u,\gamma,\psi)$
where $\gamma$ denotes nuisance functions.
We split the data into two groups
${\cal D}_0$ and ${\cal D}_1$
and we estimate $\gamma$ from ${\cal D}_0$.
The one-step estimator is
$$
\hat \psi = \hat\psi_{\rm pi} + \frac{1}{n}\sum_i \phi(U_i,\hat\gamma,\hat\psi_{\rm pi})
$$
where
$\hat\psi_{\rm pi} = \psi(\hat \gamma)$
is the plug-in estimator and the average is over ${\cal D}_1$.
This estimator comes from the von Mises expansion
of $\psi(\gamma)$ around a point $\overline{\gamma}$ given by
$\psi(\gamma) = \psi(\overline{\gamma}) + \int \phi(u,\overline{\gamma}) dP(u) + R$
where $R$ is the remainder.
Alternatively, we can define $\hat\psi$ as the solution to
the estimating equation
$n^{-1}\sum_i \phi(U_i,\hat\gamma,\psi)=0$.

Both estimators have second order bias
$||\hat\gamma - \gamma||^2$.
Under appropriate conditions,
both estimators satisfy
$\sqrt{n}(\hat\psi - \psi)\rightsquigarrow N(0,\tau^2)$
where $\tau^2 = \E[\phi^2(U,\gamma,\psi)]$.
The key condition 
for this central limit theorem to hold is that
$||\hat\gamma - \gamma||^2 = o_P(n^{-1/2})$
which holds under standard smoothness assumptions.
For example,
if $\gamma$ is in a Holder class of smoothness $s$,
then an optimal estimator $\hat\gamma$ satisfies
$||\hat\gamma - \gamma||^2 = O_P( n^{-2s/(2s+d)}) = o_P(n^{-1/2})$
when
$s> d/2$.
The plugin estimator has first order bias
$||\hat\gamma - \gamma||$
which will never be $o_P(n^{-1/2})$.

The usual confidence interval is
$\hat\psi \pm z_{\alpha/2} {\rm se}$
where
${\rm se}^2 = \hat\tau^2/n$ and
$\hat\tau^2 = n^{-1}\sum_i \phi^2(U_i,\hat\gamma)$.
But we find that this often underestimates
the standard error.
Instead, we use a different approach described in Section \ref{section::confidence}.
We consider three different estimators
for the nuisance functions
$\mu(z)$ and $\nu(z)$:
(i) linear, (ii) additive and (iii) random forests.

\subsection{Estimating $\psi_L$}
\label{section::psi}

\cite{williamson2021nonparametric}
found the
efficient influence function for $\psi_L$.
However,
in \cite{williamson2020unified}
the authors note that
one can avoid having to
use the influence function by rewriting $\psi_L$ as
$$
\psi_L = \E[ (Y-\mu(Z))^2] - \E[ (Y-\mu(X,Z))^2].
$$
It is easy to check
that the corresponding plugin estimator
$$
\hat\psi_L = \frac{1}{n}\sum_i (Y_i - \hat\mu(Z_i))^2 - 
\frac{1}{n}\sum_i (Y_i - \hat\mu(X_i,Z_i))^2
$$
already has second order bias 
$O(||\hat \mu - \mu||^2)$
so that
using the influence function is unnecessary.

\subsection{Estimating $\psi_0$}

We first derive the efficient, nonparametric estimator
of $\psi_0$ and then
we discuss some issues.
Recall that
$U =(X,Y,Z)$.

\begin{theorem}\label{theorem::psi0}
Let
$\psi_0 = \psi_0(\mu,p) = \int\int (\mu(x,z)-\mu_0(z))^2 p(x)p(z) dx dz$.
The efficient influence function is
\begin{align*}
\phi(U,\mu,p) &=
\int(\mu(x,Z) - \mu_0(Z))^2 p(x) dx  +
\int(\mu(X,z) - \mu_0(z))^2 p(z) dz\\
& \ +
2 \frac{p(X)p(Z)}{p(X,Z)}(\mu(X,Z)-\mu_0(Z))(Y - \mu(X,Z)) - 2\psi(p).
\end{align*}
In particular,
we have the following von Mises expansion.
Let $(\overline \mu,\overline p)$ be arbitrary and let
$(\mu,p)$ denote the true functions.
Then
$$
\psi_0( \mu, p) = 
\psi_0(\overline{\mu},\overline{p}) +
\int\int \phi(u,\overline\mu,\overline p) dP(u) + R
$$
where the remainder $R$ satisfies
\begin{align*}
R &=
O( ||\overline p_X-p_X|| \times ||\overline \delta- \delta||) + 
O( ||\overline p_Z-p_Z|| \times ||\overline \delta-\delta||) +
O( ||\overline p_X - p_X|| \times ||\overline p_Z - p_Z||)+
O( ||\overline \delta - \delta||^2)
\end{align*}
and
$\delta = \mu(x,z) - \mu_0(z)$.
Hence, if
$||\overline p_X-p_X|| = o_P(n^{-1/4})$,
$||\overline p_Z-p_Z|| = o_P(n^{-1/4})$,
$||\overline \delta - \delta|| = o_P(n^{-1/4})$
then
$\sqrt{n} R = o_P(1)$.
\end{theorem}

The one-step estimator
is
$$
\hat\psi_0 = \psi_0(\hat\mu,\hat p) +
\frac{1}{n}\sum_i \phi(U_i,\hat\mu,\hat p).
$$
The estimator from solving the estimating equation is
$\hat\psi = (2n)^{-1}\sum_i L(U_i,\hat\mu,\hat p)$
where
\begin{align}\nonumber
L(U,\mu,p) &=
\int(\mu(x,Z) - \mu_0(Z))^2 p(x) dx  +
\int(\mu(X,z) - \mu_0(z))^2 p(z) dz\\
& \ +
2 \frac{p(X)p(Z)}{p(X,Z)}(\mu(X,Z)-\mu_0(Z))(Y - \mu(X,Z)).
\label{eq::L}
\end{align}

\begin{corollary}
Suppose that
$||\hat p - p|| = o_P(n^{-1/4})$ and
$||\hat \mu - \mu|| = o_P(n^{-1/4})$.
When $\psi_0\neq 0$,
for either of the two estimators above,
$$
\sqrt{n}(\hat \psi_0 - \psi_0)
\rightsquigarrow N(0,\sigma^2)
$$
where
$\sigma^2 = \E[\phi^2(U,\mu,p)]$.
\end{corollary}

In our implementation,
we estimate 
$p(x,z), p(x), p(z)$ with kernel density estimators.
We estimate integrals with respect to the densities
by sampling from the kernel estimators.
Specifically,
$$
\hat\mu_{*}(z) = \frac{1}{N}\sum_{j=1}^N \hat \mu(X_j^*,z)\ \ \ \ 
{\rm where\ } X_1^*,\ldots, X_N^* \sim \hat p(x).
$$
Similarly,
$\int (\mu(X,z) - \hat \mu_0(z))^2 p(z)$ is estimated by
$$
\frac{1}{N}\sum_j (\mu(X,Z_j^*) - \hat \mu_0(Z_j^*))^2 \ \ \ \ {\rm where\ } Z_1^*,\ldots, Z_N^* \sim \hat p(z)
$$
and
$\int (\mu(x,Z) - \hat \mu_0(Z))^2 p(x)$ is estimated by
$N^{-1}\sum_j (\mu(X_j^*,z) - \hat \mu_0(z))^2.$
Thus
\begin{align*}
\hat\psi_0 &=
\frac{1}{2n}\sum_i L(U_i,\hat\mu,\hat p)\\
&=
\frac{1}{nN}\sum_i \sum_j 
\left(\hat\mu(X_j^*,Z_i)- 
\frac{1}{N}\sum_{s=1}^N \hat \mu(X_s^*,z)\right)^2 +
\frac{1}{nN}\sum_i \sum_j 
\left(\hat\mu(X_i,Z_j^*)- 
\frac{1}{N}\sum_{s=1}^N \hat \mu(X_i,Z_s^*)\right)^2 \\
&\ +
\frac{2}{n}\sum_i 
\frac{ \hat p(X_i)\hat p(Z_i)}{ \hat p(X_i,Z_i)}
\left(\hat\mu(X_i,Z_i)- \frac{1}{N}\sum_j \hat\mu(X_j^*,Z_i)\right)
(Y_i - \hat \mu(X_i,Z_i)).
\end{align*}

\bigskip

{\bf Finite Sample Problems.}
In principle, $\hat\psi_0$ is fully efficient.
In practice, $\hat\psi_0$ can behave poorly
as we now explain.
One of the terms in the von Mises remainder is
$||\hat\mu_0(z)- \mu_0(z)||^2$.
Now
$\mu_0(z) =\int\mu(x,z)p(x)dx$.
When $X$ and $Z$ are highly correlated,
there will be a large set $A_z$ of $x$ values, where
there are no observed data and so
$\hat\mu_0(z)$ will be quite far from
$\mu_0(z)$ because
$\hat\mu(x,z)$ must suffer large bias or variance (or both)
over that region.
This is known as extrapolation error.
For this reason we now consider
alternative versions of $\psi_0$.\footnote{
Readers familiar with causal inference
will recognize that, formally,
$\mu_0(z)$ is the average treatment effect if
we think of $Z$ as a treatment and $X$ as a confounder.
But the role of treatment and confounder is switched with the treatment
being the multivariate vector $Z$.
The difficulty in estimating $\mu_0(z)$
when $X$ and $Z$ are highly correlated
is known as the overlap problem in causal inference \citep{d2021overlap}.
}


\subsection{Estimating $\psi_1$}

Recall that
$\psi_1 = \E[ (\mu(X,V)-\mu(V))^2]$
where $V=(Z_j:\ |\rho(X,Z_j)| \leq t)$ for some $t$.
We take 
$\rho(X,Z_j) = \sum_{i=1}^g |\rho(X_i,Z_j)|$
where
$\rho(X_i,Z_j)$ is the Pearson correlation.
We use $t = .5$ in our examples.
For simplicity 
we assume that the values
$\rho(X,Z_j)$ are distinct.
In this case
$P(\hat V = V)\to 1$
as $n\to \infty$ where
$\hat V=(Z_j:\ |\hat\rho(X,Z_j)| \leq t)$ and
the randomness of $\hat V$ can be ignored asymptotically
and $\psi_1$ can be estimated in the same way as $\psi_L$ with
$\hat V$ replacing $Z$.

\begin{lemma}\label{lemma::psi1}
If $||\hat\mu(x,v)-\mu(x,v)||=o_P(n^{-1/4})$ 
and $\psi_1 \neq 0$ then
$\sqrt{n}(\hat\psi_1 - \psi_1)\rightsquigarrow N(0,\tau^2)$.
\end{lemma}

\subsection{Estimating $\psi_2$}

Consider the partially linear model
$Y = \beta^T X +  f(Z) + \epsilon$.
Then
$\mu_0(z) = \int \mu(x,z) p(x) dx = \beta^T m_X + f(z)$
where $m_X = \E[X]$ and so
$$
\psi_2 \equiv \int \int (\mu(x,z) -\mu_0(z))^2 p(x)p(z) dx dz = \beta^T \Sigma_X \beta
$$
and
$\beta = \E[ (Y-\mu(Z))(X-\nu(Z))]/\E[(X-\nu(Z))^2]$.

The efficient influence function
for $\psi_2$ is
$$
\phi = 2 \beta^T \Sigma_X \phi_\beta + \beta^T ( (X-m_X)(X-m_X)^T ) \beta - \psi_2
$$
where
$$
\phi_\beta = \Sigma_X^{-1} (X-\nu(Z)) \Bigl\{ (Y-\mu(Z)) - (X-\nu(Z))^T \beta) \Bigr\}
$$
and we have the von Mises expansion
$\psi_2 (\mu,\nu,\beta,\Sigma_X)  = 
\psi_2(\overline{\mu},\overline{\nu},\overline{\beta},\overline{\Sigma}_X) + 
\int \phi(u,\overline{\mu},\overline{\nu},\overline{\beta},\overline{\Sigma}_X) dP + R$
where the remainder $R$ satisfies
\begin{align*}
R&=
O(||\mu(\overline{P})-\mu(P)|| \times ||\nu(\overline{P})-\nu(P)||) +
O(||{\rm vec}(\Sigma_X(\overline{P})) - {\rm vec}(\Sigma_X(P))||^2) \\
&\ \ \ +
O(||\beta(\overline{P})-\beta(P)||^2)+
O(||\beta(\overline{P})-\beta(P)|| \times ||{\rm vec}(\Sigma_X(\overline{P})) - {\rm vec}(\Sigma_X(P))||).
\end{align*}
We omit the calcuation of the influence function and remainder
as they are standard.
Hence, if
$||\mu(\overline{P})-\mu(P)|| \times ||\nu(\overline{P})-\nu(P)||) = o(n^{-1/2})$,
$||\beta(\overline{P})-\beta(P)|| = o(n^{-1/4})$,
and
$||{\rm vec}(\Sigma_X(\overline{P})) - {\rm vec}(\Sigma_X(P))|| = o(n^{-1/4})$,
then $\sqrt{n}R = o(1)$.
It is easy to verify that
$||\beta(\overline{P})-\beta(P)|| = 
O(||\mu(\overline{P})-\mu(P)|| \times ||\nu(\overline{P})-\nu(P)||)$
and so
$\psi_2$ satisfies the double robustness property, namely,
that the bias involves the product of two quantities.
It suffices to estimate either $\mu$ or $\nu$ accurately
to get a consistent estimator.

The one-step estimator is given by
$$
\hat\psi_2 = \frac{1}{n}\sum_i \hat\beta^T (X_i - \hat\mu(Z_i))(X_i - \hat\mu(Z_i))^T \hat\beta +
\frac{2}{n}\sum_i \hat\beta^T \hat\Sigma_X \phi_\beta(X_i,Z_i)
$$
where
$$
\hat\beta=
\Biggl\{\frac{1}{n}\sum_i (X_i - \hat\nu(Z_i))(X_i - \hat\nu(Z_i))^T\Biggr\}^{-1}
\frac{1}{n}\sum_i (X_i - \hat \nu(Z_i))(Y_i - \hat \mu(Z_i))
$$
and the sums are over ${\cal D}_1$.

\subsection{Estimating $\psi_3$}
\label{section::est3}

Consider the
partially linear model with
interactions:
$$
Y = \beta^T X + \sum_{j=1}^g \sum_{k=1}^h \gamma_{jk} X_j Z_k + f(Z) + \epsilon.
$$
Define
$$
\Theta =
\left[
\begin{array}{cccc}
\beta_1 & \gamma_{11} & \cdots & \gamma_{1h}\\
\vdots  & \vdots  & \vdots  & \vdots\\
\beta_g & \gamma_{g1} & \cdots & \gamma_{gh}
\end{array}
\right],\ \ \ \ 
\mathbb{W} =
\left[
\begin{array}{cccc}
X_1 & X_1 Z_1 & \cdots & X_1 Z_h\\
\vdots  & \vdots  & \vdots  & \vdots\\
X_g  & X_{g}Z_1 & \cdots & X_g Z_h
\end{array}
\right] =
X \ \tilde{Z}^T
$$
where $\tilde Z^T = (1, Z^T)$.
Then we can write
$$
Y = \theta^T W + f(Z)+\epsilon
$$
where 
$\theta = {\rm vec}(\Theta)$ and
$W = {\rm vec}(\mathbb{W}) = {\rm vec}(X \tilde Z^T) = \tilde Z \otimes X$.

\begin{lemma}\label{lemma::psi3}
We have
$$
\theta = 
\Biggl\{
\E\bigl[\tilde Z \tilde Z^T \otimes (X-\nu(Z))(X-\nu(Z))^T \bigr]\Biggr\}^{-1}
\E\Biggl[ \biggl(Y-\mu(Z)\biggr)\  \biggl(\tilde Z \otimes (X-\nu(Z))\biggr)\Biggr]
$$
and under this model, $\psi_0$ is equal to
$\psi_3 = \theta^T \Omega \theta$ where
$$
\Omega = \Sigma_X \otimes \E[ \tilde Z \tilde Z^T] =
\Sigma_X \otimes \left(
\begin{array}{cc}
1 & m_Z^T \\
m_Z & \Sigma_Z + m_Z m_Z^T
\end{array}
\right),
$$
$m_Z = \E[Z]$ and
$\Sigma_Z = {\rm Var}[Z]$.
The efficient influence function
for $\psi_3$ is
\begin{equation}\label{eq::phi3}
\phi = 2 \theta^T \Omega \phi_\theta + \theta^T \mathring{\Omega} \theta-\psi_3
\end{equation}
where
$$
\phi_\theta = \Bigl\{\E[ R_{XZ}R_{XZ}^T]\Bigr\}^{-1}
R_{XZ}(R_Y - R_{XZ}^T \theta),
$$
$R_Y = Y-\mu(Z)$,
$R_{XZ} = {\rm vec}[(X-\nu(Z)) \tilde Z^T]$,
$$
\mathring{\Omega} =
\Biggl\{
[(X - m_X)(X - m_X)^T -\Sigma_X] \otimes
\left[
\begin{array}{cc}
1 & m_Z^T\\
m_Z & \Gamma
\end{array}
\right]
\Biggr\} +
\Biggl\{
\Sigma_X \otimes
\left[
\begin{array}{cc}
0 & (Z-m_Z)^T\\
Z-m_Z & \mathring{\Gamma}
\end{array}
\right]
\Biggr\},
$$
(the influence function of $\Omega$)
$\Gamma = \Sigma_Z + m_Z m_Z^T$,
and
$$
\mathring{\Gamma} = (Z - m_Z)(Z - m_Z)^T - \Sigma_Z + m_Z (Z-m_Z)^T + (Z-m_Z) m_Z^T
$$
(the influence function of $\Gamma$).
\end{lemma}


Then
$\psi_3 (u,\theta,\Omega)  = \psi_3(u,\overline{\theta},\overline{\Omega}) + 
\int \phi(u,\overline{\theta},\overline{\Omega}) dP(u) + R$
where the remainder $R$ satisfies
\begin{align*}
R&=
O(||\theta(\overline{P})-\theta(P)||^2)+
O(||{\rm vec}(\Omega(\overline{P})) - {\rm vec}(\Omega(P))||^2) \\
&\ \ \ +
O(||\theta(\overline{P})-\theta(P)|| \times ||{\rm vec}(\Omega(\overline{P})) - {\rm vec}(\Omega(P))||).
\end{align*}
Thus if
$||\theta(\overline{P})-\theta(P)|| = o(n^{-1/4})$ and
$||{\rm vec}(\Omega(\overline{P})) - {\rm vec}(\Omega(P))|| = o(n^{-1/4})$ 
then $\sqrt{n}R = o(1)$.
Again, we have the double robustness property.

\bigskip

The sample estimate of $\theta$ is
$\hat\theta = (\mathbb{R}_{XZ}^T\mathbb{R}_{XZ})^{-1}\mathbb{R}_{XZ}^T \mathbb{R}_Y$
where
the $i^{\rm th}$ row of
$\mathbb{R}_{XZ}$ is
${\rm vec}[(X_i - \hat\nu(Z_i)) \tilde Z_i^T]$
and $\mathbb{R}_Y(i) = Y_i - \hat\mu(Z_i)$.
Let
$\hat\Omega$ be the sample version of $\Omega$.
The one-step estimator is
$$
\hat\psi_3 = 
\frac{1}{n}\sum_i \hat\theta^T \hat\phi_\Omega (U_i) \hat\theta +
\frac{2}{n}\sum_i \hat\theta^T \hat\Omega \phi_\theta(U_i)
$$
where the sums are over ${\cal D}_1$.

\subsection{Confidence Intervals}
\label{section::confidence}

Now we describe the construction of the confidence intervals
using a method we refer to as $t$-Cross.
Let $\psi$ denote a generic parameter.
We combine two ideas:
cross-fitting 
\citep{newey2018cross}
and $t$-inference \citep{ibragimov2010t}.
Here are the steps:

\begin{enumerate}
\item Divide the data into $B$ disjoint sets
${\cal D}_1,\ldots, {\cal D}_B$;
we take $B=5$ in the examples.

\item 
Estimate the nuisance functions using
all the data except ${\cal D}_j$ and compute
$\hat\psi_j$ on ${\cal D}_j$.
\item Let $\overline{\psi} = B^{-1}\sum_{j=1}^B \hat\psi_j$.
When $\psi\neq 0$, each $\hat\psi_j$ is asymptotically Normal so that
$\overline{\psi}$ is asymptotically $t_{B-1}$.
\item
The confidence interval is
$$
\overline{\theta}\pm t_{B-1,\alpha/2}\  {\rm se}
$$
where
${\rm se}^2 = (s^2/B + c^2/n)$
where $s^2 = (B-1)^{-1}\sum_{j=1}^B (\hat\theta_j - \overline{\theta})^2$.
\end{enumerate}

Note that
$s^2$ is an unbiased estimate of the variance
of $\hat\psi$ which does not depend on the
accuracy of the estimated influence function.

\section{Simulations}
\label{section::simulations}

In this section,
we compare the behavior of the different parameters
in some synthetic examples.
For each example,
we estimate all the parameters
$\psi_L,\psi_0,\psi_1,\psi_2,\psi_3$.
To estimate the parameters
we need to estimate the nuisance functions
$\mu(z)$ and $\nu(z)$.
As mentioned above, we consider
three approaches to estimating these functions:
linear models,
additive models and random forests.
For the additive models
we use the R package \texttt{mgcv}.
For random forests
we use the R package \texttt{grf}.
We always use the default settings making no attempt
to tune the methods to achieve good coverage.

\bigskip

{\bf Example 1.}
We start with a very simple scenario
where
$Y = 2X + \epsilon$,
$\epsilon \sim N(0,1)$,
$Z_1 = \delta X + \xi$,
$\xi \sim N(0,1)$,
and
$(Z_2,\ldots,Z_5)\sim N(0,I)$.
Figure \ref{fig::example1}
shows the coverage as a function
of the correlation between $X$ and $Z_1$.
As expected, $\psi_L$ has poor coverage
as the correlation increases.
The parameter $\psi_0$ partially corrects the correlation bias
while the other parameters do a much better job.

\bigskip

{\bf Examples 2-5.}
Now we consider four multivariate examples.
In each case, $n=10,000$, $h=5$ and
$\epsilon \sim N(0,1)$.
The distributions are defined as follows:

{\em Example 2:}
$X$ is standard Normal,
$Z_1 = X + N(0,.4^2)$,
$(Z_2,\ldots,Z_h)$ is standard multivariate Normal.
The regression function is
$Y = 2 X^3 +  \epsilon$.

{\em Example 3:}
Here,
$Z\sim N(0,I)$,
$X_1 = 2 Z_1 + \epsilon_1$,
$X_2 = 2 Z_2 + \epsilon_2$,
$Y = 2 X_1 X_2 + \epsilon$
where
$\epsilon,\epsilon_1,\epsilon_2 \sim N(0,1)$.

{\em Example 4:}
Let $X\sim {\rm Unif}(-1,1)$,
$Z\sim {\rm Unif}(-1,1)$,
and
$Y = X^2 (X+ (7/5))+ (25/9)Z^2 + \epsilon$.
This example is from
\cite{williamson2021nonparametric}.
Our coverage for $\psi_L$ is similar but slightly less than that in
\cite{williamson2021nonparametric}
but we are using a different nonparametric estimator.

{\em Example 5:}
$X\sim N(0,1)$,
$Z_1 = X + N(0,.4^2)$
$(Z_2,\ldots,Z_d)\sim N(0,I)$
and
$Y = 2 X^2 + X Z_1 + \epsilon$.

\bigskip

In examples 2,4 and 5,
we replaced $X$ with 
orthogonal polynomials
$b_1(X),b_2(X), b_3(X)$.

The results 
from 100 simulations are summarized in Table \ref{table::table2}.
Figure \ref{fig::examples} shows the average
of the left and right endpoints of the confidence intervals.
The first thing to notice is that no method does
uniformly well.
Estimating variable importance well
is surprisingly difficult.
Generally, we find that $\psi_3$ works best.
However, it does poorly in two cases:
in Example 5,  with linear regressions,
and in Example 2 using random forests.
$\psi_0$ rarely does well. Apparently,
the functional is too difficult to estimate nonparametrically.
$\psi_1$ works well in a few cases,
but is not reliable enough in general.
Similar behavior occurs for $\psi_2$.
Except for a few cases, $\psi_L$ never does well.
This is not unexpected due to the correlation bias.
However, it should be noted that
these methods
are likely all doing well
in the sense of covering
the value of $\psi$ in the projected model.
For example, when using
linear models for $\mu$ and $\nu$,
we are really estimating
the value of $\psi$ for the projection
of the distribution onto the space of linear models.
The parameter estimate may capture useful information
even if it is not estimating $\psi_0$.

\begin{figure}
\begin{center}
\begin{tabular}{c}
\includegraphics[scale=.5]{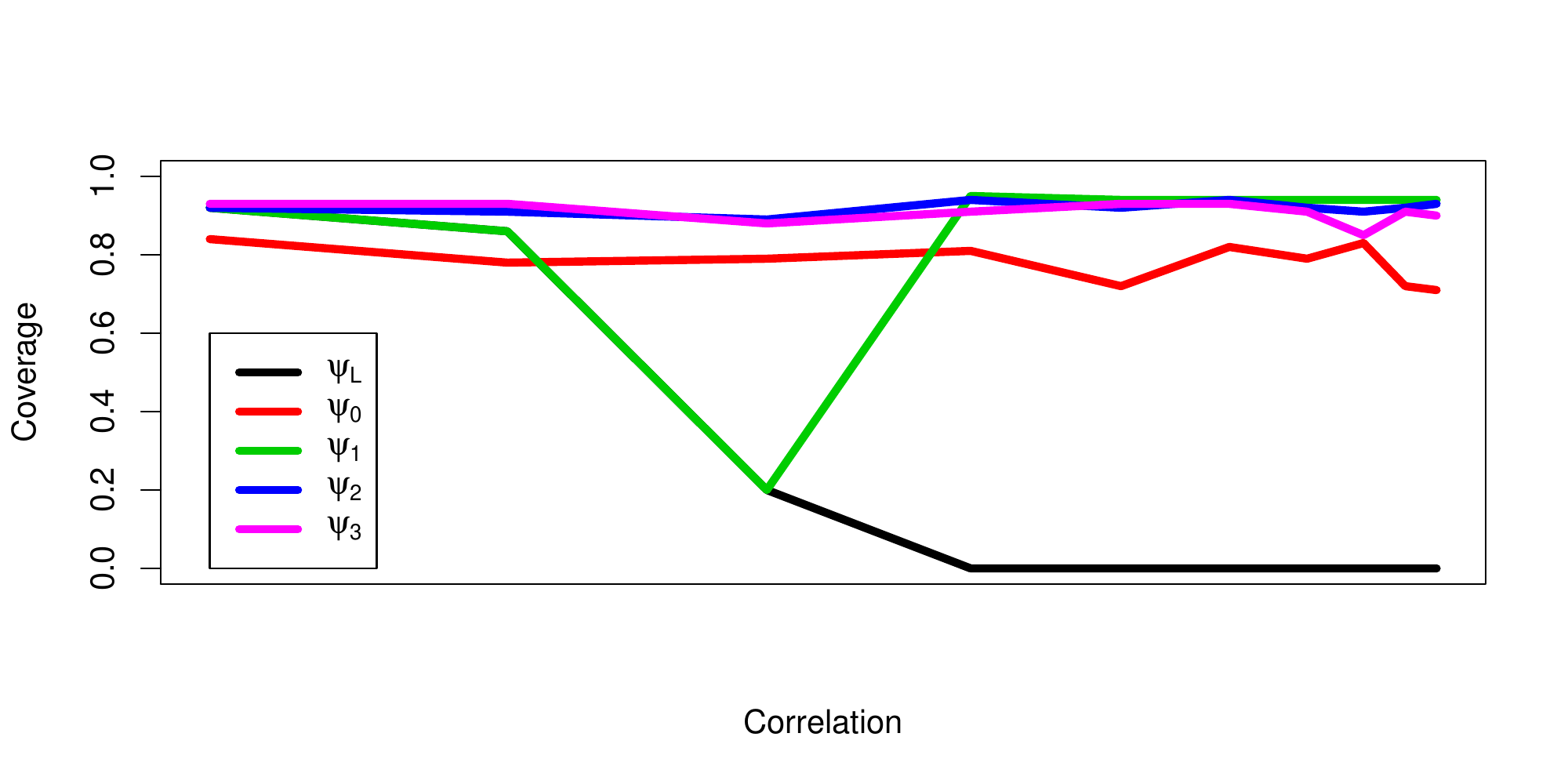}\\
\includegraphics[scale=.5]{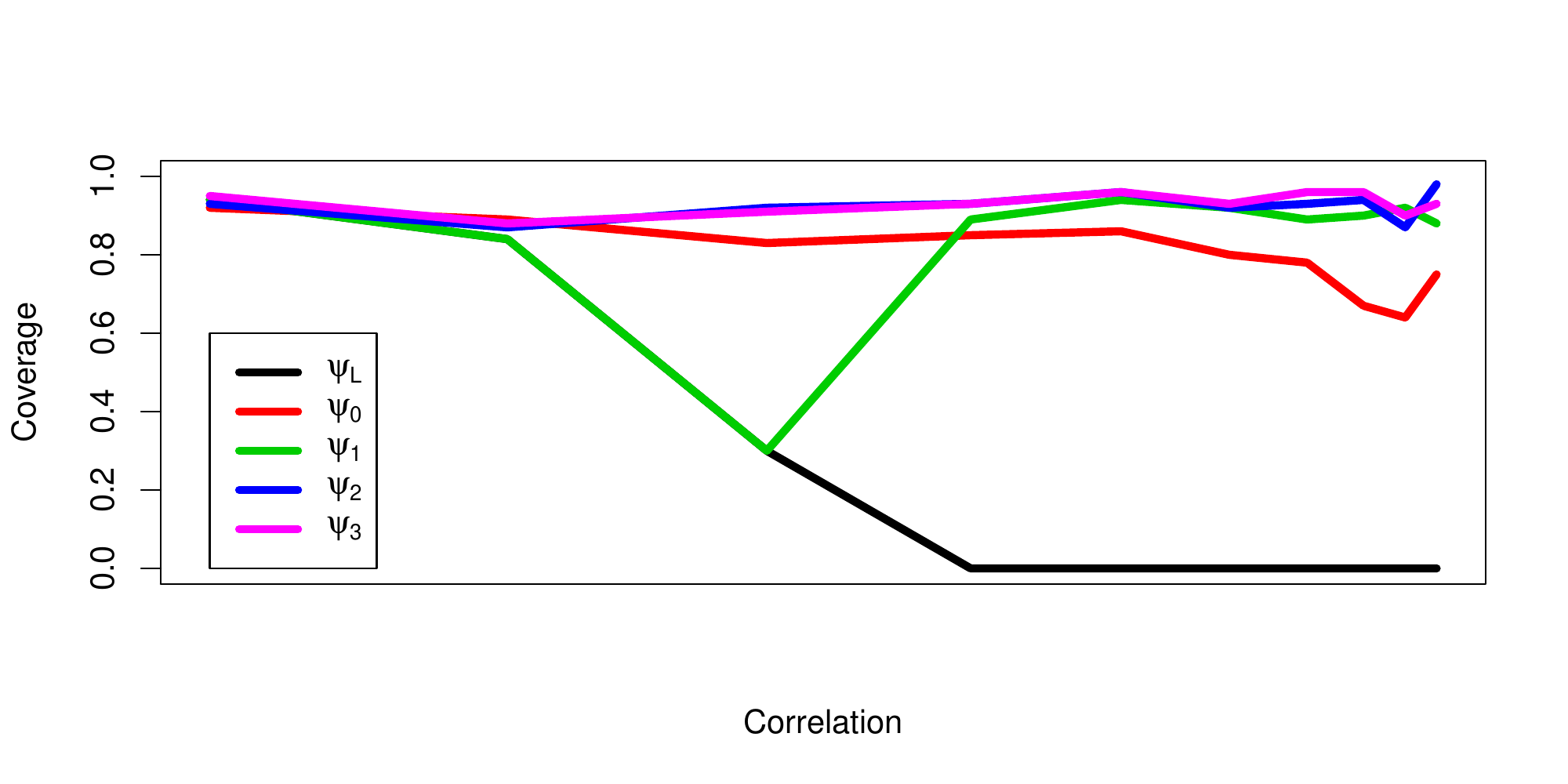}\\
\includegraphics[scale=.5]{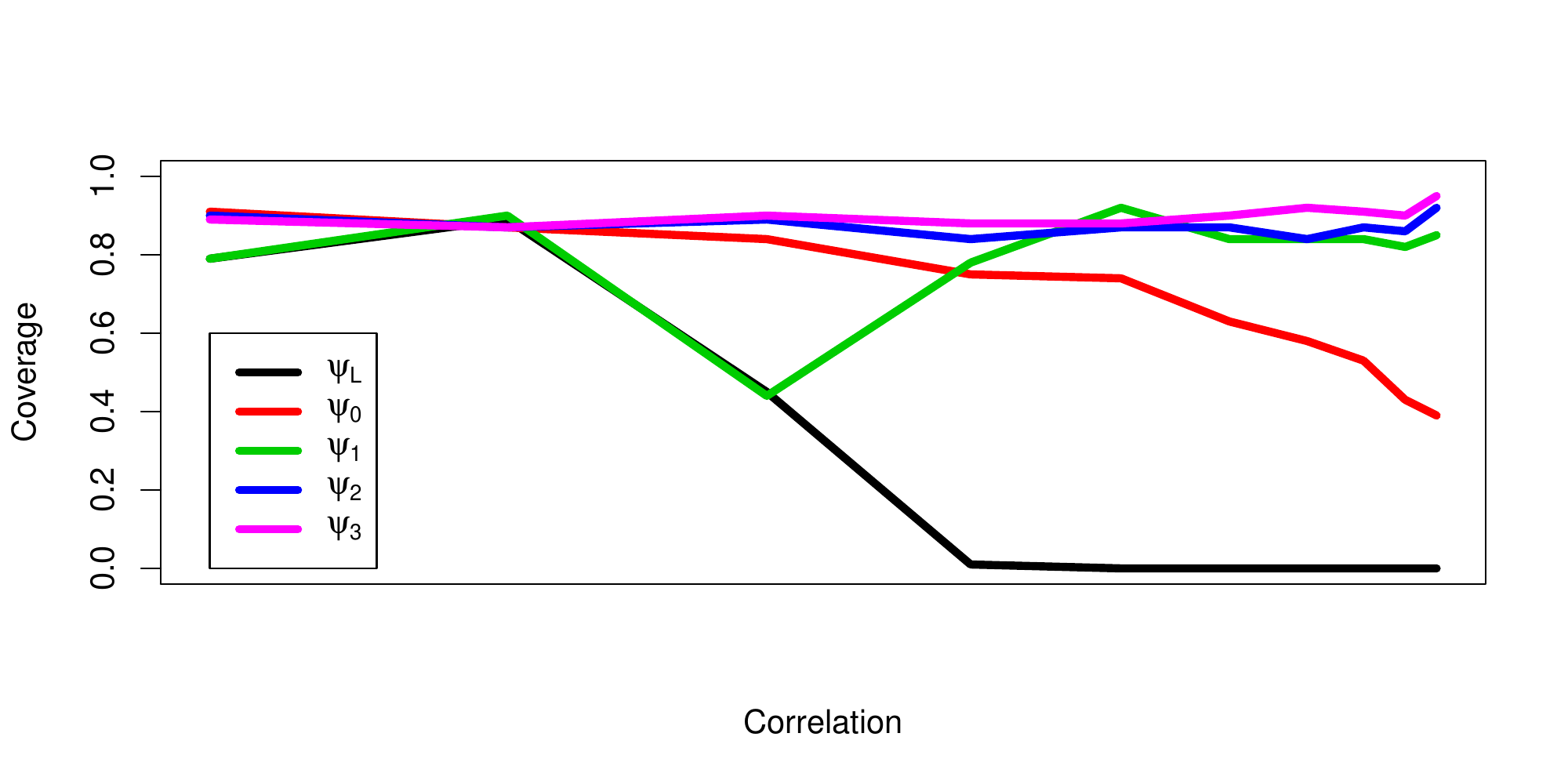}
\end{tabular}
\end{center}
\caption{Example 1: Coverage
as a function of correlation.
Top left: linear. Top right: additive. Bottom left: forests.}
\label{fig::example1}
\end{figure}

\begin{figure}
\begin{center}
\begin{tabular}{ccc}
\includegraphics[scale=.35]{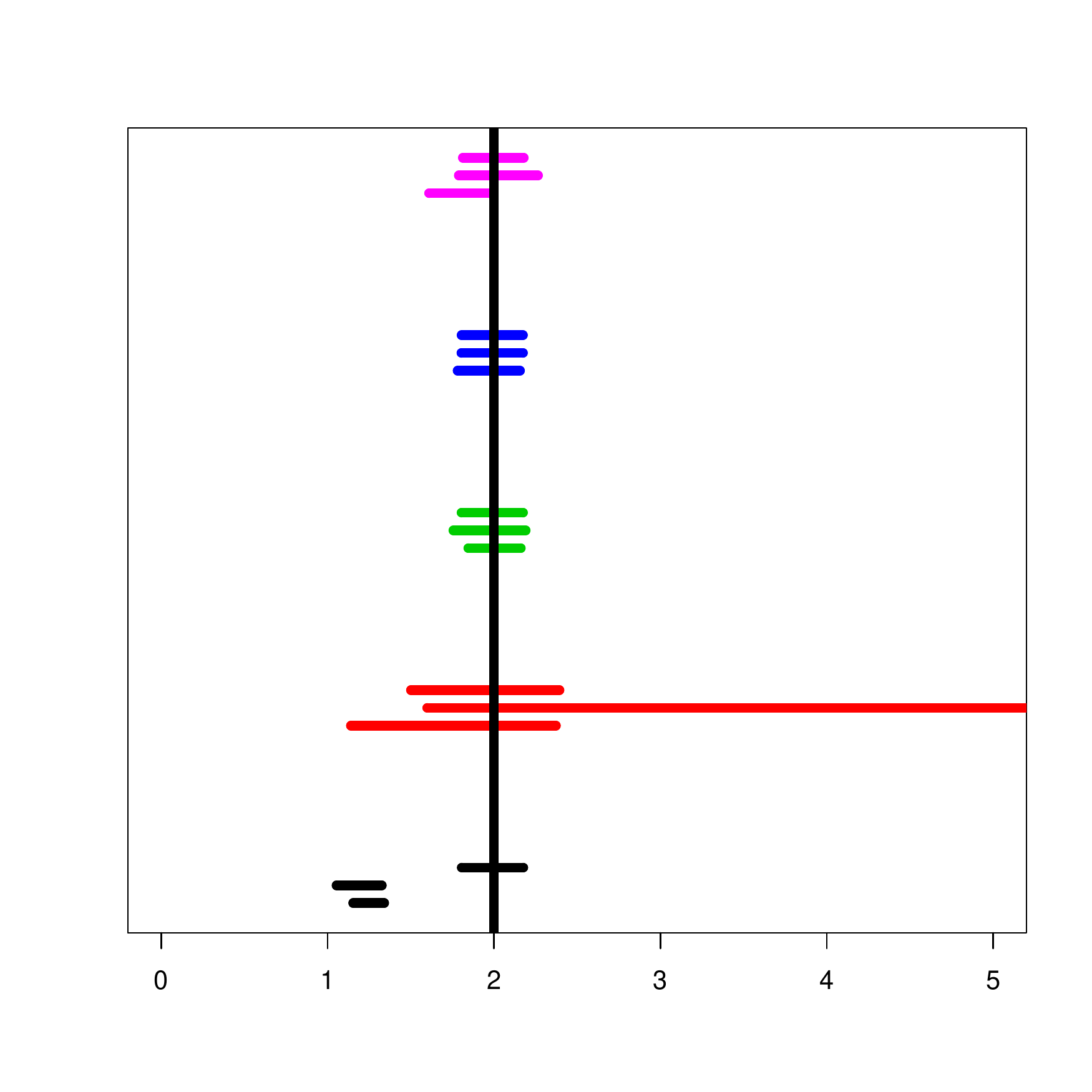}&
\includegraphics[scale=.35]{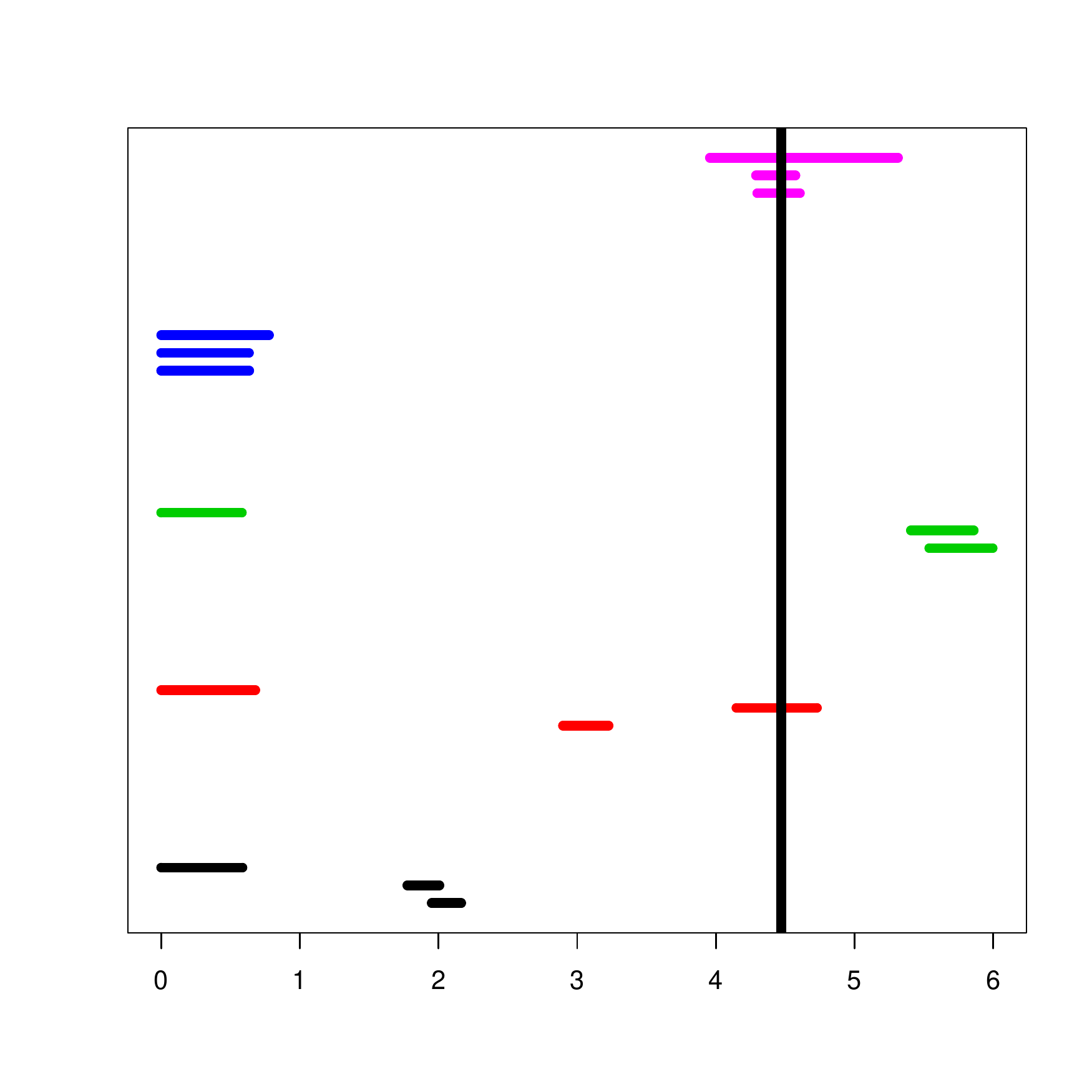}\\
\includegraphics[scale=.35]{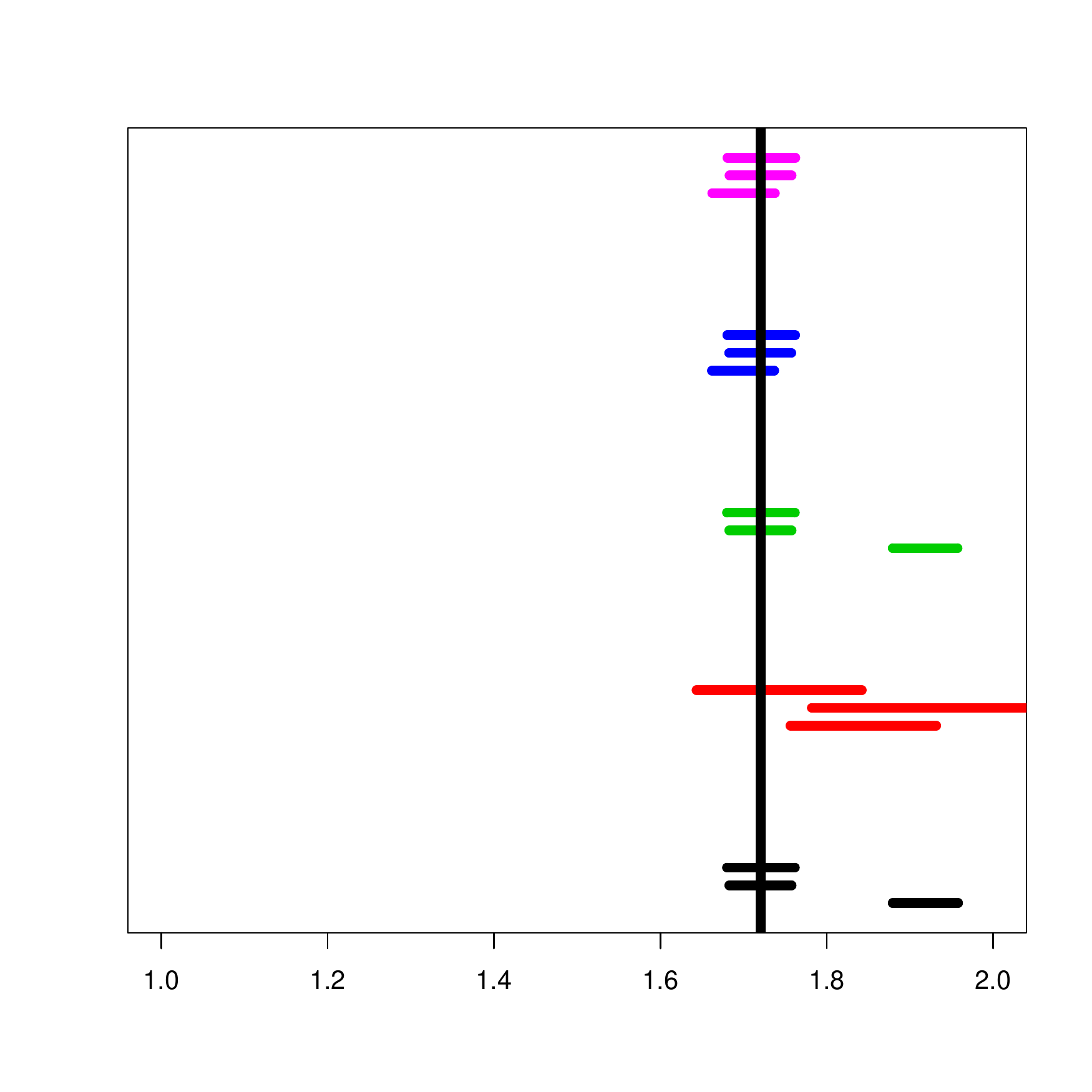}&
\includegraphics[scale=.35]{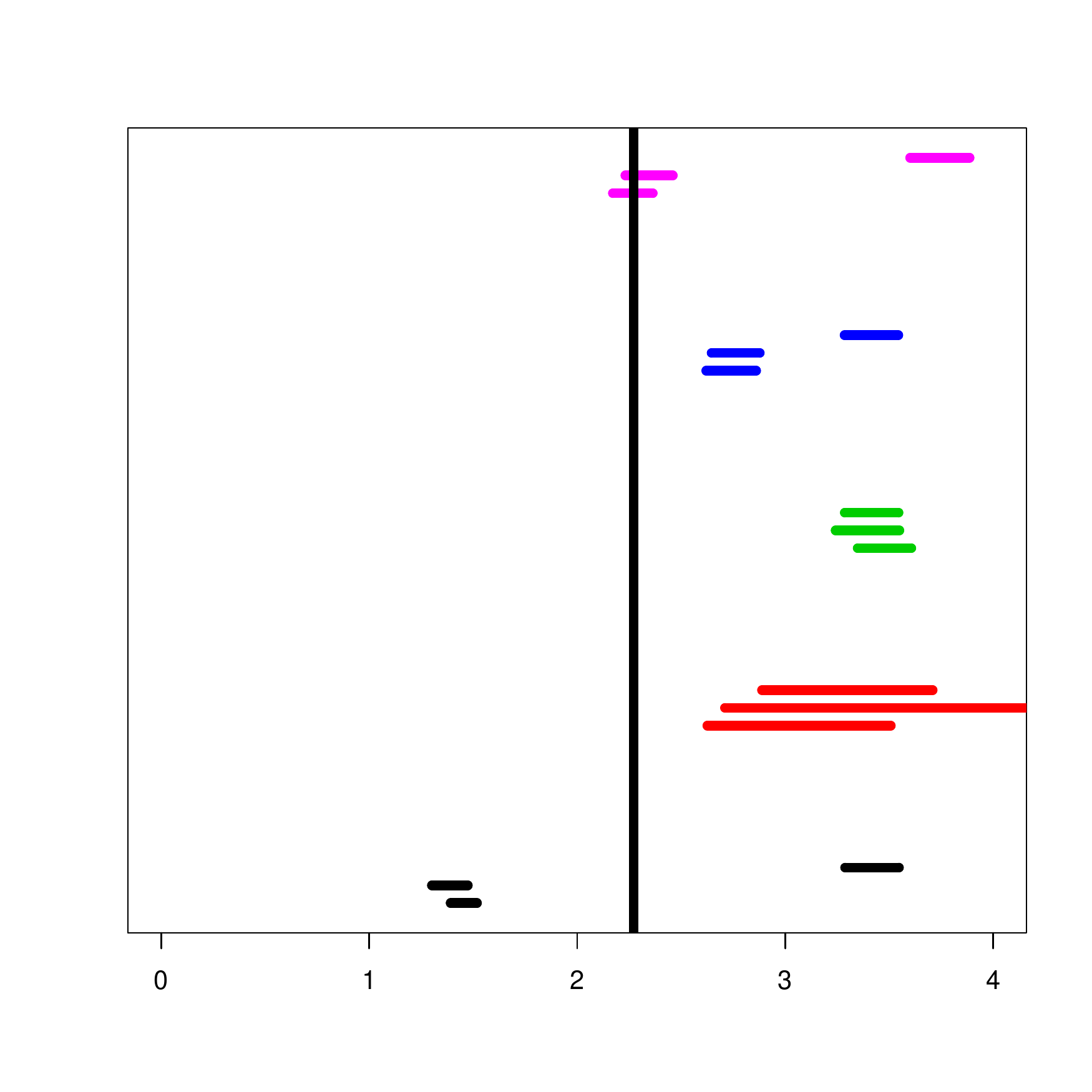}&
\includegraphics[scale=.35]{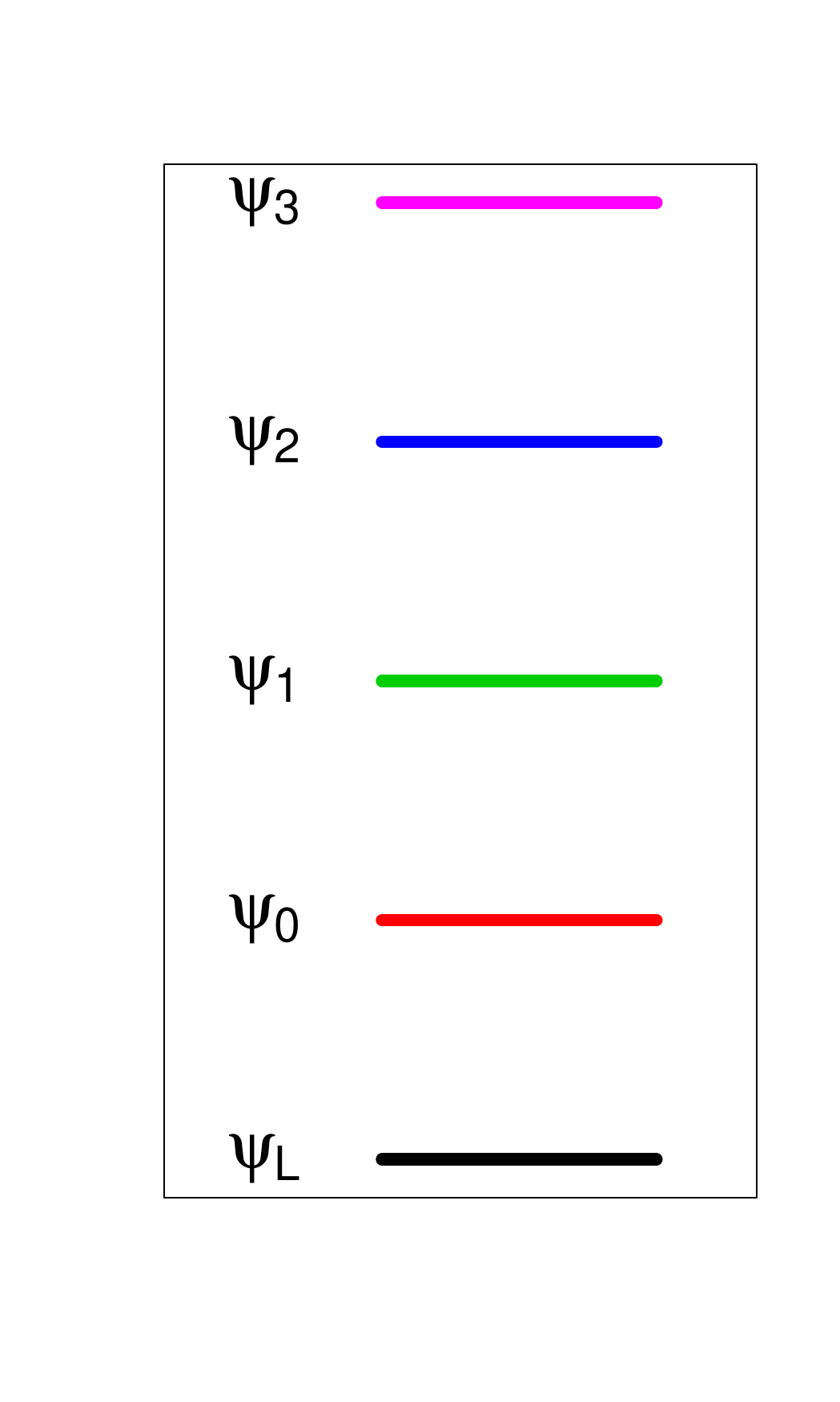}
\end{tabular}
\end{center}
\caption{The average of the left and right endpoints of the confidence intervals
over 100 simulations for Examples 2,3,4,5.
For each group of three line segments, the top is based on linear models,
the middle is based on additive models
and the bottom is based on random forests.}
\label{fig::examples}
\end{figure}

\begin{table}
\begin{center}
\begin{tabular}{c|cccccc|cccccc|ccccc}
 &   \multicolumn{5}{c}{Linear} & & \multicolumn{5}{c}{Additive} & & \multicolumn{5}{c}{Forest}\\
          &  $\psi_L$ & $\psi_0$ & $\psi_1$ & $\psi_2$ & $\psi_3$ &  & 
$\psi_L$ & $\psi_0$ & $\psi_1$ & $\psi_2$ & $\psi_3$ &  & 
$\psi_L$ & $\psi_0$ & $\psi_1$ & $\psi_2$ & $\psi_3$\\  
\hline
Example 2 &   1 & 0.84 &   1 &   1 & 1.00 & &  0.00 & 0.79 & 1.00 & 1.00 & 0.97 & &  0.00 & 0.75 & 1.00 & 0.99 & 0.30\\
Example 3 &   0 & 0.00 &   0 &   0 & 0.99 & &  0.00 & 0.88 & 0.00 & 0.00 & 0.92 & &  0.00 & 0.00 & 0.00 & 0.00 & 0.91\\
Example 4 &   1 & 0.87 &   1 &   1 & 1.00 & &  0.98 & 0.20 & 0.98 & 0.98 & 0.98 & &  0.00 & 0.21 & 0.00 & 0.86 & 0.85\\
Example 5 &   0 & 0.01 &   0 &   0 & 0.00 & &  0.00 & 0.83 & 0.00 & 0.00 & 0.85 & &  0.00 & 0.05 & 0.00 & 0.00 & 1.00\\
\end{tabular}
\end{center}
\caption{\em Coverage results for Examples 2,3,4 and 5.
Overall, $\psi_3$ performs best in these examples.
But when linear regressions are used, $\psi_3$ fails.
For random forests, $\psi_3$ does poorly in Example 2.
The most robust behavior is given by the additive model.}
\label{table::table2}
\end{table}

\section{Other Issues}
\label{section::other}

In this section we discuss
two further topics:
other variable importance parameters,
and Shapley values.

\subsection{Other Parameters}

We have focused on LOCO in this paper
but there are many other variable importance parameters
all of which can be estimated in a manner similar
to the methods in this paper.
\cite{samarov1993exploring}
suggested
$\psi = \int (\partial \mu(x,z)/\partial x)^T(\partial \mu(x,z)/\partial x)  dP$.
This parameter is not subject to correlation bias.
Estimating derivatives can be difficult
but in the semiparametric case,
$\psi$ takes a simple form.
In the partially linear model
we have $\psi = ||\beta||^2$
and in the 
partially linear model with interactions
(\ref{eq::model2})
we have
$$
\psi = ||\beta||^2 + 2\beta G^T m_Z + G^T \Sigma_Z G
$$
where
$G_{jk} = \gamma_{jk}$.

Another parameter is inspired by causal inference.
If we viewed $X$ as a treatment and $Z$ as confounding variables,
then (under some conditions) the causal effect,
that is the mean of $Y$ had $X$ been set to $x$,
is given by
Robins' $g$-formula 
$g(x) = \int \mu(x,z) dP(z)$.
We could then define $\psi$
as the variance 
${\rm Var}[g(X)]$
or the average squared derivative of 
$\int (\partial g(x)/\partial x)^T(\partial g(x)/\partial x)dP$.
These parameters do not suffer from correlation bias.
Now
${\rm Var}[g(X)]$ equals
$\beta^T \Sigma_X \beta$ under the partially linear model and is
$(\beta + \Gamma m_Z)^T\Sigma_X (\beta + \Gamma m_Z)$
under the partially linear model with interactions.
Using the derivative,
in the partially linear model
we get
$\psi = ||\beta||^2$ and
in partially linear model with interactions
we get
$$
\psi = ||\beta||^2 + 2\beta \Gamma^T m_Z + \Gamma^T m_Z m_Z^T \Gamma.
$$

The nonparametric partial correlation
is defined by 
$$
\rho = 
\frac{\E[ (Y-\mu(Z)) (X-\nu(Z))]}
{\sqrt{\E (Y-\mu(Z))^2 \E (X-\nu(Z))^2}}.
$$
Under $p_0$ we get a decorrelated version
$$
\rho_0 = 
\frac{\E_0[ (Y-\mu_0(Z)) (X-\nu_0(Z))]}
{\sqrt{\E_0 (Y-\mu_0(Z))^2 \E_0 (X-\nu_0(Z))^2}}=
\frac{\int\int (\mu(x,z)-\mu_0(z))(x-m_X)p(x)p(z) dx dz}
{\sigma_X \sqrt{\int \int\int (y-\mu_0(z))^2 p(y|x,z)p(x)p(z)}}.
$$
More detail about $\rho_0$ are in the appendix.

\subsection{Shapley Values}

A method for defining
variable importance
that has attracted
much attention lately
is based on Shapley values 
\citep{messalas2019model,
aas2019explaining,
lundberg2016unexpected,
covert2020understanding,
fryer2020shapley,
covert2020improving,
israeli2007shapley,
mase2019explaining,
benard2021shaff}.
This is an idea from
game theory where
the goal is to define the
importance of each player
in a cooperative game.
While Shapley values can be useful
in some settings, for example, computer experiments 
\citep{owen2017shapley}
we argue here that
Shapley values do not solve the decorrelation issue and
LOCO or decorrelated LOCO may be
preferable for routine regression problems.
However, this is an active area of research
and the issue is far from settled.
Shapley values may indeed have some other advantages.

The Shapley value is defined as follows.
Suppose we have covariates
$(Z_1,\ldots,Z_d)$
and that we want to measure
the importance of $Z_j$.
For any subset $S\subset\{1,\ldots, d\}$
let
$Z_S = (Z_j:\ j\in S)$ and let
$\mu(S) = \E[Y|Z_S]$.
The Shapley value for $Z_j$ is
$$
s_j=\frac{1}{d!}\sum_\pi [ V(S_j^+(\pi)) - V(S_j(\pi))]
$$
where the sum is over of permutations of
$(Z_1,\ldots, Z_d)$,
$S_j(\pi)$ denotes all variables 
before $Z_j$ in permutation $\pi$,
$S_j^+(\pi)=\{ S_j(\pi)\bigcup\{j\}\}$
and $V(S)$ is some measure of fit the regression model
with variables $S$.
If
$V(S) = - \E[(Y- \mu(S))^2]$,
then
$$
s_j= \frac{1}{d!}\sum_\pi \E[(\mu(S_j)-\mu(S_j^+))^2].
$$
This is just the LOCO parameter
averaged over all possible submodels.
The Shapley value 
for a group of variables
can be defined similarly.

It is clear that this parameter
is difficult to compute
and inference, while possible
\citep{williamson2020efficient}
is very challenging.
The appeal of the Shapley value
is that it
has the following nice properties:

\bigskip

(A1): $\sum_j s_j = \E[ (Y-\mu(Z))^2]$.

(A2) If $\E[(Y- \mu(S\bigcup\{i\}))^2]=\E[(Y- \mu(S\bigcup\{j\}))^2]$
for every $S$ not containing $i$ or $j$, then
$s_i = s_j$.

(A3) If we treat $\{Z_j,Z_k\}$ as one variable,
then its Shapley value $s_{jk}$ satisfies
$s_{jk} = s_j + s_k$.

(A4)
If $\E[(Y- \mu(S\bigcup\{j\}))^2]=\E[(Y- \mu(S))^2]$ for all $S$ then
$s_j = 0$.

\bigskip

However,
we see two problems
with Shapley values
in regression problems.
First, it defines variable importance with respect to
all submodels.
But most of those submodels are not
of interest.
Indeed, most of them would be a bad fit
to the data and are not relevant.
So it is not clear why we should
involve them in any definition of variable importance
or in the axioms.
(An intriguing idea might be to weight the submodels according to their
predictive value.)
Second, they succumb to correlation bias.
To see this, suppose that
$Y = \beta Z_1 + \epsilon$,
that the $Z_j$'s have variance 1 and that
they are perfectly correlated, that is,
$P(Z_j=Z_k)=1$ for every $j$ and $k$.
The Shapley value for $Z_1$ turns out to be
$s_1 = \beta^2/d$
which is close to 0 when
$d$ is large.
In contrast,
$\psi_0=\beta^2$
which seems more appropriate.
The confidence interval
for $\psi_0$ would have infinite length
since the design is singular
which also seems 
appropriate
since estimating the importance of a single variable 
among a set of perfectly correlated variables
should be an impossible 
inferential task.
For these reasons,
we feel that
decorrelated LOCO may have some advantages over Shapley values.

\section{A Balancing Approach}
\label{section::balancing}

In this section,
we briefly outline a different
approach to decorrelation
based on the idea of balancing
from causal inference
\citep{imai2014covariate,fong2018covariate,ning2020robust}.
The idea is to write the decorrelated density
as
$p(x)p(z) = w(x,z) p(x,z)$ where
$w(x,z) = p(z)/p(z|x)$.
If we knew $w$,
we could simply apply LOCO
using weights
$W_i = w(X_i,Z_i)$.
Now let
${\cal H} = \{h_1,\ldots, h_k\}$ be a set of functions
where
$h_j(x,z) = f_j(x)g_j(z)$.
Then
$$
\mu_j \equiv \E[f_j(X)] \E[g_j(Z)] =
\int\int h_j(x,z) p(x)p(z) dx dz =
\int\int w(x,z) dP(x,z).
$$
Now we estimate $\mu_j$ by
$$
\hat\mu_j = \frac{1}{n}\sum_i f_j(X_i)  \ \frac{1}{n}\sum_i g_j(Z_i) 
$$
and we estimate 
$\int\int w(x,z) dP(x,z)$ by
$$
\frac{1}{n}\sum_i W_i h_j(X_i,Z_i).
$$
This leads to the set of equations
$$
\hat\mu_j = \frac{1}{n}\sum_i W_i h_j(X_i,Z_i),\ \ \ j=1,\ldots, k.
$$
This does not completely specify
$W_1,\ldots, W_n$ but we can proceed as follows.
We choose $\hat W$ to minimize
$||W-\hat W||^2$ subject to
$n^{-1}\sum_i \hat W_i =1$ and the above moment constraints.
We define
$h_1(x,z) = 1$ so that the constraint
$n^{-1}\sum_i \hat W_i =1$ is included in the moment constraints.
The solution is
$$
\hat W = \mathbf{1} - H \hat\lambda
$$
where
$\mathbf{1}= (1,\ldots, 1)$,
$H$ is the $n\times k$ matrix with
$H_{ij} = h_j(X_i,Z_i)$,
$$
\hat\lambda = (n^{-1} H^T H)^{-1}(\overline{h}-\hat\mu)
$$
and
$\overline{h} = n^{-1}(\sum_i h_1(X_i,Z_i),\ldots, \sum_i h_k(X_i,Z_i))$.
The LOCO parameter and the regressions
$\hat \mu(x,z)$ and
$\hat\mu(z)$ are all estimated
using the weights.
Space does not permit a full examination of this approach;
we will report more details elsewhere.

\section{Conclusion}
\label{section::conclusion}

We showed that
correlation bias can be removed
from LOCO by modifying the definition
appropriately.
This leads to the parameter $\psi_0$.
As we have seen, getting valid inferences for $\psi_0$
nonparametrically is difficult
even in fairly simple examples.
This is mainly because
the parameter involves the function
$\mu_0(z) = \int \mu(x,z) p(x)dx$
which requires estimating $\mu(x,z)$ in regions
where there is little data due to the dependence between $x$ and $z$.

The easiest remedy is to remove correlated variables
as we did for $\psi_1$ but this led to disappointing behavior.
The other remedy was to use a semiparametric model for $\mu(x,z)$
which led to $\psi_2$ and $\psi_3$.
This appears to be the best approach.
We emphasize that even when the coverage for $\psi_2$ and $\psi_3$
is low, (when the semiparametric model is misspecified),
these parameters are still useful
if we interpret them
as projections.
For example,
$\psi_2$ measures the variable importance
of $X$ in the regression function
of the form $\beta x + f(z)$ that best approximates $\mu(x,z)$.
In the sense $\psi_2$ still captures part of the
variable importance.
\cite{graham2021semiparametrically}
discuss in detail the interpretation of misspecified semiparametric models.

We only dealt with low dimensional models.
The methods extends to
high dimensional models
by using the usual sparsity based estimators
for the nuisance functions $\mu(z)$ and $\nu(z)$.
We plan to explore this in future work.

Finally, we briefly discussed
the role of Shapley values
which have become popular
in the literature on variable importance.
The motivation for using Shapley values
appears to be correlation bias.
Indeed, if the variables were independent,
Shapley values would probably not be considered.
But we argued that they do not
adequately address the problem.
Instead, we believe that some form of decorrelation
might be preferred.

\section{Appendix}

In this appendix
we have proofs
and details for a few other parameters.

\subsection{Proofs}

{\bf Theorem \ref{theorem::psi0}.}
{\em Let
$\psi_0(\mu,p) = \int\int (\mu(x,z)-\mu_0(z))^2 p(x)p(z) dx dz$.
The efficient influence function is
\begin{align*}
\phi(X,Y,Z,\mu,p) &=
\int(\mu(x,Z) - \mu_0(Z))^2 p(x) dx  +
\int(\mu(X,z) - \mu_0(z))^2 p(z) dz\\
& \ +
2 \frac{p(X)p(Z)}{p(X,Z)}(\mu(X,Z)-\mu_0(Z))(Y - \mu(X,Z)) - 2\psi(p).
\end{align*}
In particular,
we have the following von Mises expansion
$$
\psi_0(\mu,p) = 
\psi_0(\overline{\mu},\overline{p}) +
\int\int \phi(x,y,z,\overline{\mu},\overline{p}) dP(x,y,z) + R
$$
where the remainder $R$ satisfies
\begin{align*}
||R|| &=
O(||\overline{p}(x,z)-p(x,z)||^2) +
O(||\overline{\mu}(x,z)-\mu(x,z)||^2) \\
&+
O(||\overline{p}(x,z)-p(x,z)|| \times ||\overline{\mu}(x,z)-\mu(x,z)||).
\end{align*}
}

\bigskip

{\bf Proof.}
To show that
$\phi(X,Y,Z,\mu,p)$ is the efficient
influence function
we verify that
$\phi(X,Y,Z,\mu,p)$ is the Gateuax derivative of $\psi$
and that it has the claimed second order remainder.
We will use the symbol $'$ to denote the Gateuax derivative
defined by
$$
\lim_{\epsilon\to 0} 
\frac{\psi_0( (1-\epsilon)P + \epsilon \delta_{XYZ}) - \psi_0(P)}{\epsilon}
$$
where
$\delta_{XYZ}$ is a point mass at $(X,Y,Z)$.
Also, let $\delta_{X}$ denote a point mass at $X$,
$\delta_{XY}$ a point mass at $(X,Y)$ etc.
Let
$w(x,z) = p(x)p(z)$.
Then
$\psi_0 = \int\int (\mu(x,z)-\mu_0(z))^2 w(x,z) dx dz.$
Now
$$
\psi' = 
\int\int (\mu(x,z)-\mu_0(z))^2 w'(x,z) dx dz +
2 \int\int w(x,z) (\mu(x,z)-\mu_0(z)) (\mu'(x,z) - \mu_0'(z))dx dz
$$
First, note that
$w'(x,z)= p(x)(\delta_Z(z)-p(z)) + p(z) (\delta_X(x)-p(x))$.
Next
$$
\mu(x,z) = \int y p(y|x,z) dy = \int y \frac{p(x,y,z)}{p(x,z)} dy
$$
and
$$
\mu_\epsilon(x,z) =  \int y \frac{p(x,y,z) + \epsilon (\delta_{XYZ}-p(x,y,z))}{p(x,z) + \epsilon (\delta_{XZ}-p(x,z))}dy
$$
So
\begin{align*}
\mu'(x,z) &= \int y \Biggl\{
\frac{p(x,z) (\delta_{XYZ}-p(x,y,z)) - p(x,y,z)(\delta_{XZ}-p(x,z))}{p^2(x,z)}
\Biggr\}dy\\
&=
\frac{Y}{p(X,Z)}I(x=X,z=Z)-\mu(x,z)-
\frac{\mu(x,z)I(x=X,z=Z)}{p(x,z)} + \mu(x,z)\\
&=
\frac{(Y - \mu(x,z))}{p(x,z)} I(x=X,z=Z)
\end{align*}
Now
$\mu_0(z) = \int \mu(x,z) p(x) dx$
so
\begin{align*}
\mu_0'(z) &= \int \mu(x,z) (\delta_X(x)-p(x)) dx + \int p(x) \mu'(x,z)dx\\
&=
\mu(X,z)-\mu_0(z) +
\frac{(Y - \mu(X,z))p(X)}{p(X,z)} I(z=Z)
\end{align*}
so
\begin{align*}
\phi(X,Y,Z,\mu,p) &=
\int(\mu(x,Z) - \mu_0(Z))^2 p(x) dx - \psi +
\int(\mu(X,z) - \mu_0(z))^2 p(z) dz - \psi \\
& \ +
2 w(X,Z)(\mu(X,Z)-\mu_0(Z))
\frac{(Y - \mu(X,Z))}{p(X,Z)}\\
&\ - 2
\int\int w(x,z)(\mu(x,z)-\mu_0(z))(\mu(X,z)-\mu_0(z))dx dz\\
&\ -2\frac{(Y - \mu(X,Z))p(X)}{p(X,Z)} \int w(x,Z)(\mu(x,Z)-\mu_0(Z)) dx\\
&=
\int(\mu(x,Z) - \mu_0(Z))^2 p(x) dx  +
\int(\mu(X,z) - \mu_0(z))^2 p(z) dz - 2\psi \\
& \ +
2 w(X,Z)(\mu(X,Z)-\mu_0(Z))\frac{(Y - \mu(X,Z))}{p(X,Z)}\\
&\ - 2
\int\int w(x,z)(\mu(x,z)-\mu_0(z))(\mu(X,z)-\mu_0(z))dx dz\\
&\ -2\frac{(Y - \mu(X,Z))p(X)p(Z)}{p(X,Z)} \int p(x)(\mu(x,Z)-\mu_0(Z)) dx\\
&=
\int(\mu(x,Z) - \mu_0(Z))^2 p(x) dx  +
\int(\mu(X,z) - \mu_0(z))^2 p(z) dz\\
& \ +
2 \frac{p(X)p(Z)}{p(X,Z)}(\mu(X,Z)-\mu_0(Z))(Y - \mu(X,Z)) - 2\psi(p)
\end{align*}
which has the claimed form.

Now we consider the von Mises remainder.
The remainder at $(p,\mu)$ in the direction of
$(\overline{p},\overline{\mu})$ is
$$
R = \psi(p,\mu)-\psi(\overline{p},\overline{\mu}) - 
\int \phi(u, \overline{\mu},\overline p) dP(u).
$$
Now
\begin{align*}
-R &=
\psi(\overline p,\overline{\mu}) - \psi(p,\mu) \\
&+
\int\int \overline{p}(x)p(z)  (\overline{\mu}(x,z) - \overline{\mu}_0(z))^2  dx dz +
\int\int p(x)\overline{p}(z)  (\overline{\mu}(x,z) - \overline{\mu}_0(z))^2  dx  dz\\
& +
2\int\int\int p(x,y,z) 
\frac{\overline{p}(x)\overline{p}(z)}{\overline{p}(x,z)}(\overline{\mu}(x,z)-\overline{\mu}_0(z))(y - \overline{\mu}(x,z))dx dy dz - 2\psi(p)\\
&=
\int\int \overline{p}(x)p(z)  (\overline{\mu}(x,z) - \overline{\mu}_0(z))^2  dx dz +
\int\int p(x)\overline{p}(z)  (\overline{\mu}(x,z) - \overline{\mu}_0(z))^2  dx  dz\\
& +
2\int\int\int p(x,y,z) 
\frac{\overline{p}(x)\overline{p}(z)}{\overline{p}(x,z)}(\overline{\mu}(x,z)-\overline{\mu}_0(z))(y - \overline{\mu}(x,z))dx dy dz\\
& - \psi(\overline p,\overline \mu)-\psi(p,\mu)\\
&=
\int\int \overline p(x) p(z)  (\overline{\mu}(x,z) - \overline{\mu}_0(z))^2 dxdz+
\int\int p(x) \overline p(z)  (\overline{\mu}(x,z) - \overline{\mu}_0(z))^2 dx dz\\
& +
2\int\int p (x,z) 
\frac{\overline{p}(x)\overline{p}(z)}{\overline{p}(x,z)}(\overline{\mu}(x,z)-\overline{\mu}_0(z))
(\mu(x,z) - \overline{\mu}(x,z))dx  dz\\
& - \psi(\overline p,\overline\mu)-\psi(p,\mu)\\
&=
\int\int \overline{p}(x) p(z)  (\overline{\mu}(x,z) - \overline{\mu}_0(z))^2 dxdz+
\int\int p(x) \overline{p}(z)  (\overline{\mu}(x,z) - \overline{\mu}_0(z))^2dxdz\\
& +
2\int\int \overline{p}(x)\overline{p}(z)(\overline{\mu}(x,z)-\overline{\mu}_0(z))
(\mu(x,z) - \overline{\mu}(x,z))dx  dz\\
& - \int\int \overline{p}(x)\overline{p}(z)  (\overline{\mu}(x,z) - \overline{\mu}_0(z))^2dxdz - 
\int\int p(x)p(z)  (\mu - \mu_{0})^2dxdz + 2S
\end{align*}
where
$$
S = 2 \int\int (p(x,z)-\overline{p}(x,z))
(\mu(x,z)-\overline{\mu}(x,z))\overline{p}(x)\overline{p}(z) (\overline{\mu}(x,z)-\overline{\mu}_0(z))dxdz.
$$
Now
consider the term
$m = \int\int \overline{p}(x)\overline{p}(z)(\overline{\mu}(x,z)-\overline{\mu}_0(z))
(\mu(x,z) - \overline{\mu}(x,z))dx  dz$.
We have
\begin{align*}
m &=
\int\int \overline{p}(x)\overline{p}(z)(\overline{\mu}(x,z)-\overline{\mu}_0(z))
(\mu(x,z) - \overline{\mu}(x,z))dx  dz\\
&=
\int\int \overline{p}(x)\overline{p}(z)
(\overline{\mu}(x,z)-\overline{\mu}_0(z))
(\mu(x,z) - \mu_0(z) + \mu_0(z) - \overline{\mu}_0(z) +
\overline{\mu}_0(z)  -\overline{\mu}(x,z))dx  dz\\
&=
\int\int \overline{p}(x)\overline{p}(z)(\overline{\mu}(x,z)-\overline{\mu}_0(z))(\mu(x,z) - \mu_0(z))dxdz+
\int\int \overline{p}(x)\overline{p}(z)(\overline{\mu}(x,z)-\overline{\mu}_0(z)) ( \mu_0(z) - \overline{\mu}_0(z))dxdz \\
& + 
\int\int \overline{p}(x)\overline{p}(z)(\overline{\mu}(x,z)-\overline{\mu}_0(z)) (\overline{\mu}_0(z) -\overline{\mu}(x,z))dx  dz\\
&=
\int\int \overline{p}(x)\overline{p}(z)\sqrt{\delta}\sqrt{\overline{\delta}}dxdz+ 0-
\int\int \overline{p}(x)\overline{p}(z)\overline{\delta}dxdz
\end{align*}
where
$\delta = \mu(x,z)-\mu_0(z)$ and
$\overline{\delta} = \overline{\mu}(x,z)-\overline{\mu}_0(z)$.
Hence,
\begin{align*}
-R & =
\int\int \overline{p}(x) p(z) \overline{\delta}dxdz + 
\int\int p(x)\overline{p}(z)\overline{\delta} dxdz+
2\int\int \overline{p}(x)\overline{p}(z) \sqrt{\delta}\sqrt{\overline{\delta}}dxdz\\
&-
2\int\int \overline{p}(x)\overline{p}(z) \overline{\delta}dxdz-
\int\int \overline{p}(x)\overline{p}(z) \overline{\delta}dxdz\\
&-
\int\int p(x)p(z)\delta dxdz-
\int\int \overline{p}(x)\overline{p}(z)\delta dxdz +
\int\int \overline{p}(x)\overline{p}(z)\delta dxdz\\
&=
\int\int \overline{p}(x) p(z)\overline{\delta} dxdz+
\int\int p(x) \overline{p}(z) \overline{\delta} dxdz-
\int\int \overline{p}(x) \overline{p}(z)(\sqrt{\overline{\delta}}-\sqrt{\delta})^2 dxdz\\
&-
2\int\int \overline{p}(x)\overline{p}(z)\overline{\delta} dxdz-
\int\int p(x) p(z) \delta dxdz+
\int\int  \overline{p}(x)\overline{p}(z)\delta dxdz\\
&=
\int\int (p(x) - \overline{p}(x)) \overline{p}(z)(\overline{\delta}-\delta) dx dz+
\int\int \overline{p}(x)(p(z) - \overline{p}(z))(\overline{\delta}-\delta) dxdz\\
&+
\int\int (\overline{p}(x)-p(x))(\overline{p}(z)-p(z))\delta dxdz -
\int\int \overline{p}(x) \overline{p}(z) (\sqrt{\overline{\delta}}-\sqrt{\delta})^2 dx dz.
\end{align*}
And hence
\begin{align*}
||R|| &= 
O( ||p(x)-\overline{p}(x)||\ ||\overline\delta-\delta||) + 
O( ||p(z) -\overline{p}(z)||\ ||\overline\delta-\delta||) \\
&+
O( ||p(x)-\overline{p}(x)||\ ||p(z)-\overline{p}(z)||) +
O( ||\overline{\delta} - \delta||^2)\\
&=
O(||\overline{p}(x,z)-p(x,z)||^2) +
O(||\overline{\mu}(x,z)-\mu(x,z)||^2) \\
&+
O(||\overline{p}(x,z)-p(x,z)|| \times ||\overline{\mu}(x,z)-\mu(x,z)||).\ \Box
\end{align*}

\bigskip

{\bf Lemma \ref{lemma::psi1}.}
Suppose that
$||\hat\mu(x,v)-\mu(x,v)|| = o_P(n^{-1/4})$.
Then, when $\psi_1\neq 0$,
we have that
$\sqrt{n}(\hat\psi_1-\psi_1)\rightsquigarrow N(0,\tau^2)$
for some $\tau^2$.

\begin{proof}
We have
\begin{align*}
Y_i - \hat\mu(\hat V_i) &=
(Y_i - \mu(V_i)) + (\mu(V_i)-\mu(\hat V_i)) + (\mu(\hat V_i)-\hat\mu(\hat V_i))\\
&=
(Y_i - \mu(V_i)) - (\hat V_i-V_i)^T \nabla \mu(\tilde V_i) + (\mu(\hat V_i)-\hat\mu(\hat V_i))\\
\end{align*}
for some
$\tilde V_i$ between
$V_i$ and $\hat V_i$.
Squaring, summing and letting $\epsilon_i = Y_i - \mu(V_i)$,
\begin{align*}
\frac{1}{n}\sum_i (Y_i - \hat\mu(\hat V_i))^2 &=
\frac{1}{n}\sum_i \epsilon_i^2 +
\frac{1}{n}\sum_i ((\hat V_i-V_i)^T \nabla \mu(\tilde V_i))^2 +
\frac{1}{n}\sum_i  (\mu(\hat V_i)-\hat\mu(\hat V_i))\\
& +
\frac{2}{n}\sum_i \epsilon_i (\hat V_i-V_i)^T \nabla \mu(\tilde V_i)^2 +
\frac{2}{n}\sum_i \epsilon_i (\mu(\hat V_i)-\hat\mu(\hat V_i))\\
&+
\frac{2}{n} (\hat V_i-V_i)^T \nabla \mu(\tilde V_i)(\mu(\hat V_i)-\hat\mu(\hat V_i))\\
&=
\frac{1}{n}\sum_i \epsilon_i^2 +
\frac{2}{n}\sum_i \epsilon_i (\hat V_i-V_i)^T \nabla \mu(\tilde V_i) +
\frac{2}{n}\sum_i \epsilon_i (\mu(\hat V_i)-\hat\mu(\hat V_i))+ R_n
\end{align*}
where
$R_n = O( ||\hat \delta - \delta||^2) + 
O(||\hat \mu -\mu||^2) +
O(||\hat \delta - \delta||\ ||\hat \mu -\mu||^2) = o_P(n^{-1/2})$.
The mean of the first three terms is
$\E[(Y-\mu(V))^2]$.
By a similar argument,
\begin{align*}
\frac{1}{n}\sum_i (Y_i - \hat\mu(X_i,\hat V_i))^2 &=
\frac{1}{n}\sum_i \tilde\epsilon_i^2 +
\frac{2}{n}\sum_i \tilde\epsilon_i (\hat V_i-V_i)^T \nabla \mu(X_i,\tilde V_i)\\
& +
\frac{2}{n}\sum_i \tilde\epsilon_i (\mu(X_i,\hat V_i)-\hat\mu(X_i,\hat V_i))+ \tilde R_n
\end{align*}
where
$\tilde\epsilon_i = Y_i - \mu(X_i,V_i)$,
$\tilde R_n = O( ||\hat \delta - \delta||^2) + 
O(||\hat \mu -\mu||^2) +
O(||\hat \delta - \delta||\ ||\hat \mu -\mu||^2) = o_P(n^{-1/2})$
and the mean of the first three terms is
$\E[(Y-\mu(X,V))^2]$.
The result follows from the CLT and the fact that
$\sqrt{n}(R_n + \tilde R_n)=o_P(1)$.
\end{proof}

{\bf Lemma \ref{lemma::psi3}.}
We have
that $\psi_0$ under 
the partially linear model with interactions,
is equal to
$\psi_3 = \theta^T \Omega \theta$ where
$$
\Omega = 
 \Sigma_X \otimes
\left(
\begin{array}{cc}
1 & m_Z^T \\
m_Z & \Sigma_Z
\end{array}
\right).
$$

{\bf Proof.}
Let us write
$$
\mu(x,z) = \theta^T W \equiv
\theta^T_0 X + \sum_{j=1}^h \theta_j^T X Z_j
$$
where we have written
$\theta = (\theta_0,\theta_1,\ldots, \theta_h)$
and
so
$\mu_0(z) = 
\theta_0^T m_X + \sum_{j=1}^h \theta_j^T m_X Z_j$.
Thus
\begin{align*}
(\mu(x,z)-\mu_0(z))^2 &=
\theta_0^T (X-m_X)(X-m_X)^T \theta_0 + \sum_{j=1}^h \theta_j^T (X-m_X)(X-m_X)^T Z_j^2 \theta_j \\
&\ \ +
2\sum_{j=1}^h \theta_0^T (X-m_X)(X-m_X)^T Z_j \theta_j +
2\sum_{j\neq k} \theta_j^T (X-m_X)(X-m_X)^T Z_jZ_k \theta_k
\end{align*}
and so
\begin{align*}
E_0[(\mu(x,z)-\mu_0(z))^2] &=
\theta^T \Sigma_X \theta_0 + \sum_{j=1}^h \theta_j^T \Sigma_X (\Sigma_{Z}(j,j) + m_{Z}^2(j)) \theta_j \\
&\ \ +
2\sum_{j=1}^h \theta_0^T \Sigma_X m_Z(j) \theta_j +
2\sum_{j\neq k} \theta_j^T  \theta_k (\Sigma_Z(j,k)+m_Z(j)m_Z(k))\\
&= \theta^T \Omega \theta.\ \  \Box
\end{align*}

\subsection{$\psi_L$ Under the Semiparametric Model}

Here we give the form that 
$\psi_L$ takes
under the semiparametric model.
Under the model
$\mu(x,z) = f(z) + x^T \beta(z)$,
we have
$\psi_L = \E[ \beta^T(Z) (X-\nu(Z))(X-\nu(Z))^T \beta(Z)]$
which has efficient influence function
\begin{align*}
\phi &= 2\beta(Z)^T (X-\nu(Z))(X-\nu(Z))^T V^{-1}(Z) X Y \\
&\ \ -
2\beta(Z)^T (X-\nu(Z))(X-\nu(Z))^T V^{-1}(Z) (X-\nu(Z))(X-\nu(Z))^T \beta \\
&\ \ -
\beta^T (X-\nu(Z))(X-\nu(Z))^T \beta - \psi_L.
\end{align*}
When
$\mu(x,z) = \beta^T x + \sum_{jk}\gamma_{jk}x_j z_k + f(z)$ 
then
$$
\psi_L = \theta^T (\Omega_{11} + \Omega_{12} + \Omega_{21} + \Omega_{22})
$$
where
$$
\Omega_{11} = 
\left(
\begin{array}{cc}
1 & m_Z^T \\
m_Z & \Sigma_Z + m_Z m_Z^T
\end{array}
\right) \otimes
\Sigma_X,
$$
$$
\Omega_{12} = 
\left(
\begin{array}{cc}
1 & m_Z^T \\
m_Z & \Sigma_Z + m_Z m_Z^T
\end{array}
\right) \otimes
\E[ (X-m_X) (m_X - \nu(Z))^T],
$$
$$
\Omega_{21} = 
\left(
\begin{array}{cc}
1 & m_Z^T \\
m_Z & \Sigma_Z + m_Z m_Z^T
\end{array}
\right) \otimes
\E[ (X-m_X) (m_X - \nu(Z))^T],
$$
$$
\Omega_{12} = 
\left(
\begin{array}{cc}
1 & m_Z^T \\
m_Z & \Sigma_Z + m_Z m_Z^T
\end{array}
\right) \otimes
\E[ (m_X - \nu(Z))(X-m_X)^T],
$$
$$
\Omega_{22} = 
\left(
\begin{array}{cc}
1 & m_Z^T \\
m_Z & \Sigma_Z + m_Z m_Z^T
\end{array}
\right) \otimes
\E[ (m_X - \nu(Z))(m_X-\nu(Z))^T].
$$
We omit the expression for influence function.

\subsection{Partial Correlation}

In this section,
we give the decorrelated version of the partial correlation.
Recall that
$$
\rho_0 =
\frac{\E_0[ (Y-\mu_0(Z)) (X-\nu_0(Z))]}
{\sqrt{\E_0 (Y-\mu_0(Z))^2 \E_0 (X-\nu_0(Z))^2}}=
\frac{\int\int (\mu(x,z)-\mu_0(z))(x-m_X)p(x)p(z) dx dz}
{\sigma_X \sqrt{\int \int\int (y-\mu*(z))^2 p(y|x,z)p(x)p(z)}}.
$$

\begin{theorem}
The efficient influence function for $\rho_0$ is
$$
\phi =
\frac{1}{\sqrt{\phi_2 \phi_3}}
\Biggl\{
\phi_1 -
\frac{\psi_1}{2\psi_2} \phi_2 - 
\frac{\psi_2}{2\psi_3} \phi_3\Biggr\}
$$
where, in this section, we define
\begin{align*}
\psi_1 &= \int\int (\mu(x,z)-\mu_0(z))(x-m_X)p(x)p(z) dx dz\\
\psi_2 &= \sigma_X^2\\
\psi_3 &= \int \int\int (y-\mu_0(z))^2 p(y|x,z)p(x)p(z)dx dz dy
\end{align*}
and
\begin{align*}
\phi_1 &=
\mu_0(X)(X-m)+
(X-m)\frac{Y-\mu(X,Z)}{p(X,Z)}p(X)  p(Z) + 
(X-m)p(X)\mu(X,Z) - \mu_0(z) - 2\psi_1\\
\phi_2 &= (X-m)^2 - \sigma^2_X\\
\phi_3 &= 
(Y-v(Z))^2 - \psi_3 -2 \frac{p(X)p(Z)}{p(X,Z)} (Y-\mu(X,Z)).
\end{align*}
\end{theorem}

\begin{proof}
Let us write
$\rho_0 = f(\psi_1,\psi_2,\psi_3)$
where
$f(a,b,c) = a/\sqrt{bc}$ and
\begin{align*}
\psi_1 &= \E_0[ (Y-\mu_0(Z)) (X-\nu_0(Z))]\\
\psi_2 &= \sigma_X^2\\
\psi_3 &= \int \int\int (y-\mu*(z))^2 p(y|x,z)p(x)p(z).
\end{align*}
So the influence function is
$$
f_1(\psi_1,\psi_2,\psi_3) \phi_1 + 
f_2(\psi_1,\psi_2,\psi_3) \phi_2 + 
f_3(\psi_1,\psi_2,\psi_3) \phi_3
$$
where
$f_j = \partial f/\partial \psi_j$
and
$\phi_j$ is the influence function for $\psi_j$.
Hence,
$$
\phi =
\frac{1}{\sqrt{\phi_2 \phi_3}}
\Biggl\{
\phi_1 -
\frac{\psi_1}{2\psi_2} \phi_2 - 
\frac{\psi_2}{2\psi_3} \phi_3\Biggr\}.
$$
Now
$$
\psi_1 = \int (\mu_0(x)-\psi_0)(x-m_X) p(x) dx =   \int \mu_0(x) (x-m_X) p(x) dx 
$$
where
$\mu_0(x) = \int \mu(x,z) p(z)$.
So
\begin{align*}
\phi_1 &= 
\int \mu_0(x)' (x-m_X) p(x) dx -
\int \mu_0(x)m_X' p(x) dx +
\int \mu_0(x) (x-m_X) p(x)' dx \\
&=\int \mu_0(x)' (x-m_X) p(x) dx -
\int \mu_0(x)(X-m_X) p(x) dx +
 \mu_0(X) (X-m_X) - \psi_1\\
&=
\mu_0(X) (X-m_X) +
\int \mu_0(x)' (x-m_X) p(x) dx - 2\psi_1.
\end{align*}
Now
\begin{align*}
\mu_0(x)' &= \int \mu'(x,z) p(z)dz + \mu(x,Z) - \mu_0(z)\\
&=
\int \frac{Y-\mu(x,z)}{p(x,z)}I(X=x,Z=z)  p(z)dz + \mu(x,Z) - \mu_0(z)\\
&=
I(x=X)\frac{Y-\mu(x,Z)}{p(x,Z)}  p(Z) + \mu(x,Z) - \mu_0(z).
\end{align*}
Thus
\begin{align*}
\int (x-m)p(x)
& \Biggl\{I(x=X)\frac{Y-\mu(x,Z)}{p(x,Z)}  p(Z) + \mu(x,Z) - \mu_0(z) \Biggr\}\\
&\ \ \ 
(X-m)\frac{Y-\mu(X,Z)}{p(X,Z)}p(X)  p(Z) + 
(X-m)p(X)\mu(X,Z) - \mu_0(z)
\end{align*}
So
\begin{align*}
\phi_1 &=
\mu_0(X)(X-m)+
(X-m)\frac{Y-\mu(X,Z)}{p(X,Z)}p(X)  p(Z) + 
(X-m)p(X)\mu(X,Z) - \mu_0(z) - 2\psi_1.
\end{align*}
Also
$$
\phi_2 = (X-m)^2 - \sigma^2.
$$
Now we turn to
$\psi_3 = \int \int\int (y-\mu*(z))^2 p(x,y,z)$.
Then
\begin{align*}
\phi_3 &= 
(Y-v(Z))^2 - \psi_3 -2\int p(x,y,z) (y-v(z)) v'(z)dz\\
&=
(Y-v(Z))^2 - \psi_3 -2\int p(x,z) (\mu-v(z)) v'(z)dz\\
\end{align*}
and
$$
v'(z) = \mu(X,z)-v(z)+ I(z=Z) \frac{p(X)(Y-\mu(X,z))}{p(X,z)}
$$
so that
\begin{align*}
\phi_3 &= 
(Y-v(Z))^2 - \psi_3 -2\int p(x,z) (\mu-v(z)) v'(z)dz\\
&=
(Y-v(Z))^2 - \psi_3 -2 \frac{p(X)p(Z)}{p(X,Z)} (Y-\mu(X,Z)).
\end{align*}
The remainder can be shown to be second order
in a similar way to $\psi_0$.
We omit the details.
\end{proof}

\subsection{Varying Coefficient Model}

Let
$\mu(x,z) = x^T\beta(z) + f(z)$.
In this case $\psi_0$ becomes
$\psi_4 = {\rm tr}(\Sigma_X H)$.
Define
\begin{align*}
V(z) &= {\rm Var}[X|Z=z] &\ C(z) &= {\rm Cov}[X,Y|Z=z] &\ f(z) &= \mu(z) - \nu(z)^T \beta(z)\\
\beta(z) &= V^{-1}(Z) C(z) &\ M&= \E[\beta(Z)] &\ S &= {\rm Var}[\beta(Z)].
\end{align*}

\begin{lemma}
The efficient influence function for $\psi_4$ is
$$
\phi =
{\rm tr}(\Sigma_X \phi_H) + (X-m_X)^T H (X-m) - \psi_4
$$
where
$H = \E[\beta(Z)\beta(Z)^T]$,
\begin{align*}
\phi_H &=
\beta(Z)\beta(Z)^T - H +
\beta(Z)[ Y X^T - \beta(Z)^T (X-\nu(Z))(X-\nu(Z))^T]V^{-1}(Z)\\
&\ +
V^{-1}(Z)[XY - (X-\nu(Z))(X-\nu(Z))^T \beta(Z)]\beta(Z)^T.
\end{align*}
\end{lemma}

Hence, the estimator is
$$
\hat\psi_4 =
\frac{1}{n}\sum_i 
{\rm tr} (\hat\Sigma_X \hat{\phi}_H(U_i)) +
\frac{1}{n}\sum_i (X_i - \overline{X})^T H(U_i) (X_i - \overline{X}).
$$

\bibliographystyle{imsart-nameyear}
\bibliography{paper}

\end{document}